\documentclass[11pt]{article}

\usepackage{vmargin,graphicx,amsmath,amssymb,color,subfigure,enumerate,csquotes}
\newcommand{\chg}[1]{\marginpar{\tiny #1\par}}
\usepackage{amsthm}
\newtheorem{theorem}{Theorem}[section]
\newtheorem{lemma}[theorem]{Lemma}
\newtheorem{notation}[theorem]{Notation}
\newtheorem{proposition}[theorem]{Proposition}
\newtheorem{corollary}[theorem]{Corollary}

\theoremstyle{definition}
\newtheorem{definition}[theorem]{Definition}

\def\blacksquare{
\thinspace\nobreak \vrule width 5pt height 5pt depth 0pt}
\newtheorem{remark}[theorem]{Remark}

\makeatletter

\@addtoreset{equation}{section}
\makeatother
\def\dx{\, {\rm d}}

\newcommand{\la}{\lambda}
\newcommand{\eps}{\varepsilon}
\newcommand{\dr}{\partial}

\newcommand{\R}{\mathbb{R}}

\newcommand{\dist}{\mathrm{dist}}

\newcommand{\Dir}{\mathsf{Dir}}
\newcommand{\Neu}{\mathsf{Neu}}
\newcommand{\eff}{\mathsf{eff}}

\newcommand{\curl}{\mathsf{curl}\,}
\newcommand{\F}{\mathsf{F}}

\newcommand{\A}{\mathsf{A}}
\newcommand{\B}{\mathsf{B}}
\newcommand{\app}{\mathsf{app}}
\newcommand{\appp}{\mathsf{app2}}
\newcommand{\lin}{\mathsf{lin}}
\newcommand{\defo}{\mathsf{def}}

\newcommand{\para}{\parallel}

\newcommand{\Dom}{\mathsf{Dom}}

\newcommand{\Id} {\mathsf{Id}}

\date{\small March 26, 2013}
\title{\bf Magnetic effects in curved quantum waveguides}
\author{D.~Krej{\v{c}}i{\v{r}}{\'{\i}}k\footnote{Department of Theoretical Physics, Nuclear Physics Institute ASCR, 25068 \u{R}e\u{z}, Czech Republic; 
e-mail: \texttt{krejcirik@ujf.cas.cz}} \ and \ N.~Raymond\footnote{IRMAR, Universit\'e de Rennes 1, Campus de Beaulieu, F-35042 Rennes cedex, France;
e-mail: \texttt{nicolas.raymond@univ-rennes1.fr}}}

\begin{document}
\maketitle

\begin{abstract}
\chg{}
\noindent
The interplay among the spectrum, geometry and magnetic field
in tubular neighbourhoods of curves in Euclidean spaces is investigated
in the limit when the cross section shrinks to a point.
Proving a norm resolvent convergence,
we derive effective, lower-dimensional models 
which depend on the intensity of the magnetic field and curvatures.
The results are used to establish complete asymptotic
expansions for eigenvalues.
Spectral stability properties based on Hardy-type inequalities
induced by magnetic fields are also analysed.
\end{abstract}
%

\section{Introduction}
This paper is concerned with spectral properties of 
a curved quantum waveguide when a magnetic field is applied. 
The configuration space of the waveguide is modelled
by a tube~$\Omega_{\eps}$ about an unbounded curve~$\gamma$ 
in the Euclidean space~$\R^d$, with $d \geq 2$,
where~$\eps$ is a positive parameter 
that homothetically scales the waveguide cross section
$\{\eps\tau: \tau \in \omega\}$. All along this paper $\omega\subset\R^d$ will be a bounded and simply connected domain. The quantum dynamics is governed by the Dirichlet realization
$\mathfrak{L}_{\eps,b\A}^{[d]}$ of the magnetic Laplacian 
\begin{equation}\label{Laplace}
  (-i\nabla_{x}+b\;\!\A(x))^2
  \qquad \mbox{on} \qquad
  L^2(\Omega_{\eps},\dx x)
  ,
\end{equation}
where $b>0$ is a positive parameter and $\A$ a smooth vector potential
associated with a given magnetic field~$\B$.

We are primarily interested in effective models 
for $\mathfrak{L}_{\eps,b\A}^{[d]}$ in the limit $\eps\to 0$,
which corresponds to the scaled cross section of the waveguide
shrinking to a point. 
The other parameter with which we can play is 
the intensity of the magnetic field~$b$. 
From a heuristic point of view, if~$b$ is fixed and~$\eps$ goes to zero, 
we expect that the limiting model will not depend on the magnetic field. 
Indeed, in the limit when~$\eps$ goes to zero, 
$\Omega_{\eps}$ shrinks to the curve~$\gamma$ 
and there is no magnetic field in dimension~$1$. 
However, the situation is much less clear if the parameter~$b$
is allowed to depend on~$\eps$. We shall show that the effective
model substantially depends on the smallness of~$\eps b$
and reveal thus different asymptotic regimes
which lead to distinct spectral phenomena. 

In dimensions 2 and 3, the case without magnetic field and without torsion 
is investigated in the famous paper of Duclos and Exner \cite{Duclos95} 
(see also subsequent generalizations in \cite{Duclos05, KK05, FK08} 
as well as \cite{DEK01, Carron04, LL06, LL07, LL07b, Ro12} 
where quantum layers are analysed). 
In particular, they prove that there is always discrete spectrum 
below the essential spectrum when the waveguide is not straight 
but it is straight asymptotically.
They also investigate the limit $\eps\to0$ to show that the Dirichlet Laplacian on the tube $\Omega_{\eps}$ converges in a suitable sense to the effective one dimensional operator:
\begin{equation*}
  \mathcal{L}^{{\eff}}=-\dr_{s}^2-\frac{\kappa(s)^2}{4}
  \qquad \mbox{on} \qquad
  L^2(\gamma,\dx s)
  .
\end{equation*}
In addition it is proved in \cite{Duclos95} 
that each eigenvalue of this operator
generates an eigenvalue of the initial operator $\mathfrak{L}_{\eps,0}^{[d]}$. 

In dimension~3 it is also possible to twist the waveguide 
by allowing the cross section of the waveguide to non-trivially 
rotate by an angle function~$\theta$
with respect to a relatively parallel frame of~$\gamma$
(then the velocity~$\theta'$ can be interpreted 
as a \enquote{torsion}). 
It is proved in \cite{EKK08} that, 
whereas the curvature is favourable to discrete spectrum, 
the torsion plays against it. 
In particular, the spectrum of a straight 
twisted waveguide is stable under small perturbations
(such as local electric field or bending). 
This repulsive effect of twisting is quantified in~\cite{EKK08}
(see also \cite{K08,KZ1})
by means of a Hardy type inequality. 
The interplay between the effects bending and twisting 
is illuminated in the limit $\eps\to 0$
when one reveals the effective operator \cite{BMLT07,deOliveira_2006,KS12, Grushin_2009} (see also \cite{LTW11, TW13}):
\begin{equation}\label{effective-model}
  \mathcal{L}^{\eff}=-\dr_{s}^2-\frac{\kappa(s)^2}{2}
  +C(\omega)\;\!\theta'(s)^2
  \qquad \mbox{on} \qquad
  L^2(\gamma,\dx s)
  ,
\end{equation}
where $C(\omega)$ is a positive constant whenever $\omega$ is not a disk or annulus. 

Writing~\eqref{Laplace} in suitable curvilinear coordinates
(see~\eqref{L3} below),
one may notice similarities in the appearance of 
the torsion and the magnetic field in the coefficients 
of the operator and it therefore seems natural to
ask the following question: 
\begin{center}
\enquote{Does the magnetic field act as the torsion ?} 
\end{center}
It turns out that the question of the limit $\eps\to 0$ 
in the presence of a magnetic field is investigated in~\cite{Grushin_2008} 
where a model operator of the form \eqref{effective-model} is derived in the case when the waveguide is periodic with respect to~$s$. In particular, the restriction to $\gamma$ of the vector potential appears in the effective model 
(see \cite[Eq.~(2.3)]{Grushin_2008}) 
and cannot be completely gauged out due to the periodicity.

Our previous remark on the fact that the magnetic field 
disappears in the limit $\eps\to 0$ if~$b$ is kept fixed
leads us to the study up to which extent the similarity 
can be justified on the level of various smallness regimes of~$\eps b$.
In our paper we derive an appropriate effective dynamics~$\mathcal{L}^{\eff}$
in each of the regimes. Especially we establish that, 
as soon as~$b$ is of order~$\eps^{-1}$, new effective operators appear and display a competition between the magnetic field and the torsion. Moreover, in the critical regime, we establish complete asymptotic expansions in~$\eps$
for eigenvalues below the essential spectrum. This regime $\eps b\sim1$ is critical in the sense that $\eps b\gg 1$ is a semiclassical regime (which is beyond the scope of this paper).

Two-dimensional waveguides with magnetic field 
were previously investigated in \cite{EK05}. 
The authors prove that, as in the case with torsion, there is a Hardy type inequality 
in the straight waveguide:
\begin{equation}\label{Hardy0}
\int_{\R\times(-1,1)} 
|(-i\nabla+b\A)\psi|^2\dx x
-\int_{\R\times(-1,1)} |\psi|^2\dx x
\geq C\int_{\R\times(-1,1)} \frac{|\psi|^2}{1+x_1^2}\dx x
,
\end{equation}
for $\psi\in H^1_{0}(\R\times(-1,1))$,
where $C$ is a positive constant 
whenever the magnetic field is not identically zero,
and they use it to ensure the stability of the spectrum 
under small perturbations. 
In this paper we extend the inequality~\eqref{Hardy0}
to any dimension~$d$ and investigate the dependence
of the constant~$C$ on the magnetic field. 
We also establish the spectral stability properties
in the full generality.

The organization of the paper is as follows. 
In the forthcoming Section 2 we present our main results in detail.
The remaining sections are devoted to proofs.
The effective models together with asymptotic expansions 
for eigenvalues in dimension two (respectively, three) 
are derived in Section~\ref{2D} (respectively, Section~\ref{3D}).
The magnetic Hardy inequality in arbitrary dimension
and associated spectral stability are established in Section~\ref{rep}.

\section{Main results}
A precise statement of our results requires to start
with some technical prerequisites.

\subsection{Magnetic field}
Let $\A: \Omega_\eps \to \R^d$ denote a smooth vector field,
which will play the role of our magnetic potential.
In the canonical coordinates of~$\R^d$, 
denoted here by $x=(x_{1},\cdots,x_{d})$,
the vector potential induces a $1$-form 
$$\xi_{\A}=\sum_{j=1}^d \A_{j}(x) \dx x_{j}.$$
Its exterior derivative is given by:
$$\sigma_{\B}=d\xi_{\A}=\sum_{1\leq i<j\leq d}\B_{ij}\dx x_{i}\wedge \dx x_{j}, 
\qquad\mbox{where}\qquad
 \B_{ij}=\dr_{x_{i}}\A_{j}-\dr_{x_{j}}\A_{i}
$$
are coefficients of the magnetic tensor.
In dimension~$2$, respectively~$3$, 
the magnetic field~$\B$ can be identified with the scalar, 
respectively the vector:
$$
  \B=\B_{12},
  \qquad\mbox{respectively}\qquad 
  \B=(\B_{23},-\B_{13}, \B_{12}).
$$
We will assume throughout the whole paper that~$\B$ has compact support. 

\begin{remark}
Since $\omega$ is simply connected, all our results will only depend on the magnetic field~$\B$. This is convenient in order to compare the geometric effect of torsion with a physical quantity.
If it were not the case, one should slightly adapt our proofs and we would get results involving the vector potential $\A$.
\end{remark}

\subsection{Two-dimensional waveguides} 
Let us consider a smooth and injective curve $\gamma$: $\R\ni s\mapsto \gamma(s)$ which is parameterized by its arc length $s$. The normal 
to the curve at $\gamma(s)$ is defined as the unique unit vector~$\nu(s)$ 
such that $\gamma'(s)\cdot\nu(s)=0$ and $\det(\gamma',\nu)=1$. We have the relation $\gamma''(s)=-\kappa(s)\nu(s)$ where $\kappa(s)$ denotes the algebraic curvature at the point $\gamma(s)$. 

We can now define standard tubular coordinates. We consider:
$$\R\times(-\eps,\eps)\ni(s,t)\mapsto\Phi(s,t)=\gamma(s)+t\nu(s).$$
We always assume 
\begin{equation}\label{not-overlap}
\mbox{$\Phi$ is injective}
\qquad\mbox{and}\qquad 
\eps\sup_{s\in\R} |\kappa(s)|<1.
\end{equation}
Then it is well known (see \cite{KK05}) that $\Phi$ defines a smooth diffeomorphism from $\R\times(-\eps,\eps)$ 
onto the image $\Omega_{\eps}=\Phi(\R\times(-\eps,\eps))$,
which we identify with our waveguide.
 
Up to changing the gauge, 
the Laplace-Beltrami expression of $\mathfrak{L}_{\eps,b\A}^{[2]}$ in these coordinates is given by (see \cite[App.~F]{FouHel10}):
$$\mathfrak{L}^{[2]}_{\eps,b\mathcal{A}}=(1-t\kappa(s))^{-1}(i\dr_{s}+b\mathcal{A}_{1})(1-t\kappa(s))^{-1}(i\dr_{s}+b\mathcal{A}_{1})-(1-t\kappa(s))^{-1}\dr_{t}(1-t\kappa(s))\dr_{t},$$
with the gauge:
$$\mathcal{A}(s,t)=(\mathcal{A}_{1}(s,t),0)  ,
\quad \mathcal{A}_{1}(s,t)=\int_{0}^t (1-t'\kappa(s))\B(\Phi(s,t'))\dx t'.$$
We let:
$$m(s,t)=(1-t\kappa(s))^{-1/2}.$$
The self-adjoint operator $\mathfrak{L}^{[2]}_{\eps,b\mathcal{A}}$ on $L^2(\R\times(-\eps,\eps),m \dx s\dx t)$ is unitarily equivalent to the self-adjoint operator on $L^2(\R\times(-\eps,\eps),\dx s\dx t)$:
$$\mathcal{L}^{[2]}_{\eps,b\mathcal{A}}=m^{-1}\mathfrak{L}^{[2]}_{\eps,b\A}m.$$
Introducing the rescaling
\begin{equation}\label{rescaling}
t=\eps \tau
,
\end{equation}
we let:
$$\mathcal{A}_{\eps}(s,\tau)=(\mathcal{A}_{1,\eps}(s,\tau),0)=(\mathcal{A}_{1}(s,\eps\tau),0)$$
and denote by $\mathcal{L}^{[2]}_{\eps,b \mathcal{A}_{\eps}}$ the homogenized operator on $L^2(\R\times(-1,1),\dx s\dx\tau)$:
\begin{equation}\label{L2}
\mathcal{L}^{[2]}_{\eps, b\mathcal{A}_{\eps}}=m_{\eps}(i\dr_{s}+b\mathcal{A}_{1,\eps})m_{\eps}^2(i\dr_{s}+b\mathcal{A}_{1,\eps})m_{\eps}-\eps^{-2}\dr_{\tau}^2+V_{\eps}(s,\tau),
\end{equation}
with:
$$m_{\eps}(s,\tau)=m(s,\eps\tau),\quad V_{\eps}(s,\tau)=-\frac{\kappa(s)^2}{4}(1-\eps\kappa(s)\tau)^{-2}.$$

Henceforth we assume that the curvature~$\kappa$ has compact support. 
Recalling that also~$\B$ is supposed to be smooth and have compact support,
it is easy to verify that 
$\mathcal{L}^{[2]}_{\eps, b\mathcal{A}}$,
defined as Friedrichs extension of the operator
initially defined on $\mathcal{C}^\infty_{0}(\R\times(-\eps,\eps))$,
has form domain $H_0^1(\R\times(-\eps,\eps))$.
Similarly, the form domain of
$\mathcal{L}^{[2]}_{\eps, b\mathcal{A}_{\eps}}$
is $H_0^1(\R\times(-1,1))$.

\subsection{Three-dimensional waveguides}
The situation is geometrically more complicated in dimension~3. 
We consider a smooth curve~$\gamma$ 
which is parameterized by its arc length~$s$
and does not overlap itself. 
We use the so-called Tang frame 
(or the relatively parallel frame, see for instance \cite{KS12}) 
to describe the geometry of the tubular neighbourhood of~$\gamma$. 
Denoting the (unit) tangent vector by $T(s)=\gamma'(s)$, 
the Tang frame $(T(s), M_{2}(s), M_{3}(s))$
satisfies the relations:
\begin{eqnarray*}
T'&=&\kappa_{2} M_{2}+\kappa_{3} M_{3},\\
M_{2}'&=&-\kappa_{2} T,\\
M_{3}'&=&-\kappa_{3}T.
\end{eqnarray*}
Here~$\kappa_2$ and~$\kappa_3$ are curvature functions 
relative to the choice of the normal fields~$M_2$ and~$M_3$.  
Although the latter (and therefore the former) are not uniquely defined,
$\kappa^2=\kappa_{2}^2+\kappa_{3}^2=|\gamma''|^2$
is just the square of the usual curvature of~$\gamma$. 

Let $\theta:\R\to\R$ a smooth function (twisting). 
We introduce the map
$\Phi : \R\times(\eps\omega)\to\Omega_{\eps}$ defined by:
\begin{equation}\label{Phi}
x=\Phi(s,t_{2},t_{3})=\gamma(s)+t_{2}(\cos\theta M_{2}(s)+\sin\theta M_{3}(s))+t_{3}(-\sin\theta M_{2}(s)+\cos\theta M_{3}(s)).
\end{equation}
Let us notice that $s$ will often be denoted by $t_{1}$. 
As in dimension two, we always assume:
\begin{equation}\label{not-overlap3}
\mbox{$\Phi$ is injective}
\qquad\mbox{and}\qquad 
\eps\sup_{(\tau_{2},\tau_{3})\in\omega}(|\tau_{2}|+|\tau_{3}|)\,\sup_{s\in\R} |\kappa(s)|<1.
\end{equation}
Sufficient conditions ensuring the infectivity hypothesis
can be found in~\cite[App.~A]{EKK08}.

We define 
$\mathcal{A}= D\Phi \A(\Phi)
=(\mathcal{A}_{1}, \mathcal{A}_{2}, \mathcal{A}_{3})$, 
\begin{eqnarray*}
h&=&1-t_{2}(\kappa_{2}\cos\theta+\kappa_{3}\sin\theta)-t_{3}(-\kappa_{2}\sin\theta+\kappa_{3}\cos\theta),\\
h_{2}&=&-t_{2}\theta',\\
h_{3}&=&t_{3}\theta',
\end{eqnarray*}
and $\mathcal{R}=h_{3}b\mathcal{A}_{2}+h_{2}b\mathcal{A}_{3}$. We also introduce the angular derivative $\dr_{\alpha}=t_{3}\dr_{t_{2}}-t_{2}\dr_{t_{3}}$. We will see in Section \ref{3D} that the magnetic operator $\mathfrak{L}^{[3]}_{\eps,b\A}$ is unitarily equivalent to the operator on $L^2(\Omega_{\eps},h\dx t)$ given by:
\begin{multline}\label{L-frak3}
\mathfrak{L}^{[3]}_{\eps,b\mathcal{A}}=\sum_{j=2,3}h^{-1}(-i\dr_{t_{j}}+b\mathcal{A}_{j})h(-i\dr_{t_{j}}+b\mathcal{A}_{j})\\
+h^{-1}(-i\dr_{s}+b\mathcal{A}_{1}-i\theta'\dr_{\alpha}+\mathcal{R})h^{-1}(-i\dr_{s}+b\mathcal{A}_{1}-i\theta'\dr_{\alpha}+\mathcal{R}).
\end{multline}
By considering the conjugate operator $h^{1/2}\mathfrak{L}^{[3]}_{\eps,b\mathcal{A}} h^{-1/2}$, we find that $\mathfrak{L}^{[3]}_{\eps,b\mathcal{A}}$ is unitarily equivalent to the operator defined on 
$L^2(\R\times(\eps\omega),\dx s \dx t_2 \dx t_3)$ 
given by:
\begin{multline}
\mathcal{L}^{[3]}_{\eps,b\mathcal{A}}=\sum_{j=2,3} (-i\dr_{t_{j}}+b\mathcal{A}_{j})^2-\frac{\kappa^2}{4h^2}\\
+h^{-1/2}(-i\dr_{s}+b\mathcal{A}_{1}-i\theta'\dr_{\alpha}+\mathcal{R})h^{-1}(-i\dr_{s}+b\mathcal{A}_{1}-i\theta'\dr_{\alpha}+\mathcal{R})h^{-1/2}.
\end{multline}
Finally, introducing the rescaling
$$(t_{2}, t_{3})=\eps(\tau_{2}, \tau_{3})=\eps\tau,$$
we define the homogenized operator on $L^2(\R\times\omega,\dx s \dx\tau)$:
\begin{multline}\label{L3}
\mathcal{L}^{[3]}_{\eps,b\mathcal{A}_{\eps}}=\sum_{j=2,3} (-i\eps^{-1}\dr_{\tau_{j}}+b\mathcal{A}_{j,\eps})^2-\frac{\kappa^2}{4h_{\eps}^2}\\ 
+h_{\eps}^{-1/2}(-i\dr_{s}+b\mathcal{A}_{1,\eps}-i\theta'\dr_{\alpha}+\mathcal{R}_{\eps})h_{\eps}^{-1}(-i\dr_{s}+b\mathcal{A}_{1,\eps}-i\theta'\dr_{\alpha}+\mathcal{R}_{\eps})h_{\eps}^{-1/2},
\end{multline}
where $\mathcal{A}_{\eps}(s,\tau)=\mathcal{A}(s,\eps\tau)$, $h_{\eps}(s,\tau)=h(s,\eps\tau)$ and $\mathcal{R}_{\eps}=\mathcal{R}(s,\eps\tau)$.

Henceforth we assume that~$\kappa$ and~$\theta'$ have compact supports.
Again, it is possible to verify that the form domains of 
$\mathcal{L}^{[3]}_{\eps, b\mathcal{A}}$
and $\mathcal{L}^{[3]}_{\eps, b\mathcal{A}_{\eps}}$
are $H_0^1(\R\times(-\eps,\eps))$ and $H_0^1(\R\times(-1,1))$,
respectively.

\subsection{Limiting models and asymptotic expansions}
We can now state our main results concerning the effective models in the limit $\eps\to 0$. We will denote by $\la_{n}^\Dir(\omega)$ the $n$-th eigenvalue of the Dirichlet Laplacian $-\Delta^\Dir_{\omega}$ on $L^2(\omega)$. The first positive and $L^2$-normalized eigenfunction will be denoted by $J_{1}$. 

\begin{definition}[Case $d=2$] 
For $\delta\in(-\infty,1)$, we define:
$$\mathcal{L}^{\eff, [2]}_{\eps,\delta}=-\eps^{-2}\Delta^\Dir_{\omega}-\dr_{s}^2-\frac{\kappa(s)^2}{4}$$
and for $\delta=1$, we let:
$$\mathcal{L}^{\eff,[2]}_{\eps,1}=-\eps^{-2}\Delta^\Dir_{\omega}+\mathcal{T}^{[2]},$$
where
$$\mathcal{T}^{[2]}=-\dr_{s}^2+\left(\frac{1}{3}+\frac{2}{\pi^2}\right) \B(\gamma(s))^2-\frac{\kappa(s)^2}{4}.$$
\end{definition}

\begin{theorem}[Case $d=2$]\label{Thm-2D}
There exists $K$ such that, for all $\delta\in(-\infty,1]$, there exist $\eps_{0}>0, C>0$ such that for all $\eps\in(0,\eps_{0})$:
$$\left\|\left(\mathcal{L}^{[2]}_{\eps,\eps^{-\delta}\mathcal{A}_{\eps}}-\eps^{-2}\la_{1}^\Dir(\omega)+K\right)^{-1}-\left(\mathcal{L}^{\eff,[2]}_{\eps,\delta}-\eps^{-2}\la_{1}^\Dir(\omega)+K\right)^{-1}\right\|\leq C\eps^{1-\delta}, \mbox{ for } \delta<1$$
and:
$$\left\|\left(\mathcal{L}^{[2]}_{\eps,\eps^{-1}\mathcal{A}_{\eps}}-\eps^{-2}\la_{1}^\Dir(\omega)+K\right)^{-1}-\left(\mathcal{L}^{\eff,[2]}_{\eps,1}-\eps^{-2}\la_{1}^\Dir(\omega)+K\right)^{-1}\right\|\leq C\eps.$$
\end{theorem}
In the critical regime $\delta=1$, we deduce the following corollary providing the asymptotic expansions of the lowest eigenvalues $\la_{n}^{[2]}(\eps)$ of $\mathcal{L}^{[2]}_{\eps,\eps^{-1}\mathcal{A}_{\eps}}$.

\begin{corollary}[Case $d=2$ and $\delta=1$]\label{expansion-eigenvalues-2D}
Let us assume that $\mathcal{T}^{[2]}$ admits $N$ (simple) eigenvalues $\mu_{0},\cdots, \mu_{N}$ below the threshold of the essential spectrum. Then, for all $n\in\{1,\cdots N\}$, there exist $(\gamma_{j,n})_{j\geq 0}$ and $\eps_{0}>0$ such that for all $\eps\in(0,\eps_{0})$:
\footnote{We write $\mu(\eps)\sim\sum_{j\geq j_{0}}\mu_{j}\eps^{j}$ when for all $J\geq j_{0}$ we can find $\eps_{0}>0$ and $C>0$ such that for $\eps\in(0,\eps_{0})$: $|\mu(\eps)-\sum_{j=j_{0}}^J \mu_{j}\eps^j|\leq C\eps^{J+1}$.}
$$\la_{n}^{[2]}(\eps)\underset{\eps\to 0}{\sim}\sum_{j\geq 0} \gamma_{j,n}\eps^{-2+j},$$
with 
$$\gamma_{-2,n}=\frac{\pi^2}{4},\quad \gamma_{-1,n}=0,\quad \gamma_{0,n}=\mu_{n}.$$
\end{corollary}
Thanks to the spectral theorem, we also get the approximation of the corresponding eigenfunctions at any order (see our quasimodes in \eqref{quasi}). 

In order to present analogous results in dimension three,
we introduce supplementary notation.
The norm and the inner product in $L^2(\omega)$
will be denoted by $\|\cdot\|_{\omega}$
and $\langle\cdot,\cdot\rangle_{\omega}$, respectively.
We also use the standard notation $D_{x}=-i\dr_{x}$.

\begin{definition}[Case $d=3$] For $\delta\in(-\infty,1)$, we define:
$$\mathcal{L}^{\eff, [3]}_{\eps,\delta}=-\eps^{-2}\Delta^\Dir_{\omega}-\dr_{s}^2-\frac{\kappa(s)^2}{4}+\|\dr_{\alpha}J_{1}\|_{\omega}^2\theta'^2$$
and for $\delta=1$, we let:
$$\mathcal{L}^{\eff,[3]}_{\eps,1}=-\eps^{-2}\Delta^\Dir_{\omega}+\mathcal{T}^{[3]},$$
where $\mathcal{T}^{[3]}$ is defined by:
\begin{multline*}
\mathcal{T}^{[3]}=\langle(-i\dr_{s}-i\theta'\dr_{\alpha}-\mathcal{B}_{12}(s,0,0)\tau_{2}-\mathcal{B}_{13}(s,0,0)\tau_{3})^2 \Id(s)\otimes J_{1}, \Id(s)\otimes J_{1}\rangle_{\omega}\\
+\mathcal{B}_{23}^2(s,0,0)\left(\frac{\|\tau J_{1}\|_{\omega}^2}{4}-\langle D_{\alpha}R_{\omega},J_{1}\rangle_{\omega}\right)-\frac{\kappa^2(s)}{4},
\end{multline*}
with $R_{\omega}$ being given in \eqref{J1} and 
\begin{eqnarray*}
\mathcal{B}_{23}(s,0,0)&=&\B(\gamma (s))\cdot T(s),\\
\mathcal{B}_{13}(s,0,0)&=&\B(\gamma (s))\cdot(\cos\theta\, M_{2}(s)-\sin\theta\, M_{3}(s)),\\
\mathcal{B}_{12}(s,0,0)&=&\B(\gamma (s))\cdot(-\sin\theta\, M_{2}(s)+\cos\theta\, M_{3}(s)).
\end{eqnarray*}
\end{definition}
\begin{theorem}[Case $d=3$]\label{Thm-3D}
There exists $K$ such that for all $\delta\in(-\infty,1]$, there exist $\eps_{0}>0, C>0$ such that for all $\eps\in(0,\eps_{0})$:
$$\left\|\left(\mathcal{L}^{[3]}_{\eps,\eps^{-\delta}\mathcal{A}_{\eps}}-\eps^{-2}\la_{1}^\Dir(\omega)+K\right)^{-1}-\left(\mathcal{L}^{\eff,[3]}_{\eps,\delta}-\eps^{-2}\la_{1}^\Dir(\omega)+K\right)^{-1}\right\|\leq C\eps^{1-\delta}, \mbox{ for } \delta<1$$
and:
$$\left\|\left(\mathcal{L}^{[3]}_{\eps,\eps^{-1}\mathcal{A}_{\eps}}-\eps^{-2}\la_{1}^\Dir(\omega)+K\right)^{-1}-\left(\mathcal{L}^{\eff,[3]}_{\eps,1}-\eps^{-2}\la_{1}^\Dir(\omega)+K\right)^{-1}\right\|\leq C\eps.$$
\end{theorem}
In the same way, this theorem implies asymptotic expansions of eigenvalues $\la_{n}^{[3]}(\eps)$ of $\mathcal{L}^{[3]}_{\eps,\eps^{-1}\mathcal{A}_{\eps}}$.
\begin{corollary}[Case $d=3$ and $\delta=1$]\label{expansion-eigenvalues-3D}
Let us assume that $\mathcal{T}^{[3]}$ admits $N$ (simple) eigenvalues $\nu_{0},\cdots, \nu_{N}$ below the threshold of the essential spectrum. Then, for all $n\in\{1,\cdots N\}$, there exist $(\gamma_{j,n})_{j\geq 0}$ and $\eps_{0}>0$ such that for all $\eps\in(0,\eps_{0})$:
$$\la_{n}^{[3]}(\eps)\underset{\eps\to 0}{\sim}\sum_{j\geq 0} \gamma_{j,n}\eps^{-2+j},$$
with 
$$\gamma_{-2,n}=\la_{1}^\Dir(\omega),\quad \gamma_{-1,n}=0,\quad \gamma_{0,n}=\nu_{n}.$$
\end{corollary}
As in two dimensions, we also get the corresponding expansion for the eigenfunctions. Complete asymptotic expansions for eigenvalues in finite 
three-dimensional waveguides
without magnetic field were also previously established in 
\cite{Grushin_2009, Borisov-Cardone_2011}. Such expansions were also obtained in \cite{Grushin_2008} in the case $\delta=0$ in a periodic framework.

We refer to Sections~\ref{2D} and~\ref{3D} for proofs of the results
in the case of dimension~2 and~3, respectively. 
In agreement with the expectation mentioned in the introduction,
no magnetic effect can be observed via the limiting models 
provided that the quantity~$\eps b$ is negligible 
in the limit as $\eps \to 0$. 
On the other hand, in dimension~$2$ the magnetic field 
plays the same (repulsive) role as the torsion in dimension~$3$
provided that~$\eps b$ is of order one.
The effect of magnetic field is much more complex in dimension~$3$
in the latter regime.

\subsection{A Hardy inequality in straight magnetic waveguides} 
In dimension $2$, the limiting model (with $\delta=1$) enlightens the fact that the magnetic field plays against the curvature, whereas in dimension $3$ this repulsive effect is not obvious (it can be seen that $\langle D_{\alpha}R_{\omega},J_{1}\rangle_{\omega}\geq 0$). Nevertheless, if $\omega$ is a disk, we have $\langle D_{\alpha}R_{\omega},J_{1}\rangle_{\omega}=0$ and thus the component of the magnetic field parallel to~$\gamma$ plays against the curvature (in comparison, a pure torsion has no effect when the cross section is a disk). In the flat case ($\kappa=0$), we can quantify this repulsive effect by means of a magnetic Hardy inequality 
(see Section~\ref{rep} for the proofs).
\begin{theorem}\label{stability}
Let $d\geq 2$. Let us consider $\Omega=\R\times \omega$. 
For $R>0$, we let:
$$\Omega(R)=\{t\in\Omega : |t_{1}|< R\}.$$
Let $\A$ be a smooth vector potential such that $\sigma_{\B}$ is not zero on $\Omega(R_{0})$ for some $R_{0}>0$. Then, there exists $C>0$ such that, for all $R\geq R_{0}$, there exists $c_{R}(\B)>0$ such that, we have:
\begin{equation}\label{Hardy}
\int_{\Omega} |(-i\nabla+\A)\psi|^2-\la_{1}^\Dir(\omega)|\psi|^2\,\dx t\geq \int _{\Omega}\frac{c_{R}(\B)}{1+s^2}|\psi|^2\dx t,\quad \forall\psi\in \mathcal{C}^\infty_{0}(\Omega).
\end{equation}
Moreover we can take:
$$c_{R}(\B)=\left(1+CR^{-2}\right)^{-1}\min\left(\frac{1}{4},\la_{1}^{\Dir,\Neu}(\B,\Omega(R))-\la_{1}^\Dir(\omega)\right),$$
where $\la_{1}^{\Dir,\Neu}(\B,\Omega(R))$ denotes the first eigenvalue of the magnetic Laplacian on $\Omega(R)$, with Dirichlet condition on $\R\times\dr\omega$ and Neumann condition on $\{|s|=R\}\times\omega$.
\end{theorem}
\begin{remark}\label{rem-b}
The diamagnetic inequality (see for instance \cite[Prop. 2.1.3]{FouHel10}) implies that $\la_{1}^{\Dir,\Neu}(\B,\Omega(R))>\la_{1}^{\Dir,\Neu}(0,\Omega(R))\geq\la_{1}^{\Dir}(\omega)$. If $\B$ does not vanish on $\Omega(R)$, we have:
$$\lim_{b\to+\infty} c_{R}(b\B)=\frac{1}{4}\left(1+CR^{-2}\right)^{-1}.$$
In this sense we could say that the constant $\frac{1}{4}$, coming from the standard Hardy inequality in dimension $1$, is optimal. 
Moreover, Theorem~\ref{stability} generalizes the one of \cite{EK05} to any dimension and provides a very simple Hardy constant $c_{R}(\B)$ which explicitly displays the relation between diamagnetism and the existence a magnetic Hardy inequality. Since $\B$ is compactly supported, thanks to perturbation theory, we can show that $c(b\B)\underset{b\to 0}{\sim} c b^2$ for some positive $c$. In particular we recover the behaviour of the explicit constant of \cite{EK05}.
\end{remark}

\subsection{Spectral stability due to the magnetic field} 
The inequality of Theorem \ref{stability} can be applied to prove 
certain stability of the spectrum of the magnetic Laplacian on $\Omega$ under local and small deformations of $\Omega$. Let us fix $\eps>0$ and describe a generic deformation of the straight tube $\Omega$. 
We consider the local diffeomorphism:
$$\Phi_{\eps}(t)=\Phi_{\eps}(s,t_{2},t_{3})=(s,0,\cdots,0)+\sum_{j=2}^d (t_{j}+\eps_{j}(s))M_{j}+\mathcal{E}_{1}(s),$$
where $(M_{j})_{j=2}^d$ is the canonical basis of $\{0\}\times\R^{d-1}$. The functions $\eps_{j}$ and $\mathcal{E}_{1}$ are smooth and compactly supported in a compact set $K$. As previously we assume that $\Phi_{\eps}$ is a global diffeomorphism and we consider the deformed tube $\Omega^{\defo,\eps}=\Phi_{\eps}(\R\times\omega)$.

\begin{proposition}\label{perturbed}
Let $d\geq 2$. There exists $\eps_{0}>0$ such that for $\eps\in(0,\eps_{0})$, the spectrum of the Dirichlet realization of $(-i\nabla+\A)^2$ on $\Omega^{\defo,\eps}$ coincides with the spectrum of the Dirichlet realization of $(-i\nabla+\A)^2$ on $\Omega$. The spectrum is given by $[\la_{1}^\Dir(\omega),+\infty)$.
\end{proposition}

As we have noticed in Remark \ref{rem-b}, a large magnetic field does not increase very much $c(\B)$. Nevertheless in the large magnetic field limit (which is equivalent to a semiclassical limit with parameter $h=b^{-1}$), 
it is possible to prove a stability result which does not use the Hardy inequality.
\begin{proposition}\label{strong-b}
Let $R_{0}>0$ and $\Omega(R_{0})=\{t\in\R\times\omega : |t_{1}|\leq R_{0}\}$. Let us assume that $\sigma_\B=\dx\xi_\A$ does not vanish on $\Phi(\Omega(R_{0}))$ and that on $\Omega_{1}\setminus \Phi(\Omega(R_{0}))$ the curvature is zero. Then, there exists $b_{0}>0$ such that for $b\geq b_{0}$, the discrete spectrum of $\mathfrak{L}^{[d]}_{1,b\A}$ is empty.
\end{proposition}

\subsection{Norm resolvent convergence}
Finally, let us state an auxiliary result, 
inspired by the approach of \cite{Friedlander-Solomyak_2007},
which tells us that, in order to estimate the difference between two resolvents, 
it is sufficient to analyse the difference between the corresponding sesquilinear forms as soon as their domains are the same.
\begin{lemma}\label{NRC}
Let $\mathfrak{L}_{1}$ and $\mathfrak{L}_{2}$ be two positive self-adjoint operators on a Hilbert space $\mathsf{H}$. Let $\mathfrak{B}_{1}$ and $\mathfrak{B}_{2}$ be their associated sesquilinear forms. We assume that $\Dom(\mathfrak{B}_{1})=\Dom(\mathfrak{B}_{2})$. Assume that there exists $\eta>0$ such that for all $\phi,\psi\in\Dom(\mathfrak{B}_{1})$:
$$\left|\mathfrak{B}_{1}(\phi,\psi)-\mathfrak{B}_{2}(\phi,\psi)\right|\leq \eta \sqrt{\mathfrak{Q}_{1}(\psi)}\sqrt{\mathfrak{Q}_{2}(\phi)},$$
where $\mathfrak{Q}_{j}(\varphi)=\mathfrak{B}_{j}(\varphi,\varphi)$ for $j=1,2$ and $\varphi\in\Dom(\mathfrak{B}_{1})$.
Then, we have:
$$\|\mathfrak{L}_{1}^{-1}-\mathfrak{L}_{2}^{-1}\|\leq \eta\|\mathfrak{L}_{1}^{-1}\|\|\mathfrak{L}_{2}^{-1}\|.$$
\end{lemma}
\begin{proof}
The proof can be found in \cite[Prop. 5.3]{KS12} 
but we recall it for the convenience of the reader. Let us consider $\tilde\phi,\tilde\psi\in \mathsf{H}$. We let $\phi=\mathfrak{L}_{2}^{-1}\tilde\phi$ and $\psi=\mathfrak{L}_{1}^{-1}\tilde\psi$. We have $\phi,\psi\in\Dom(\mathfrak{B}_{1})=\Dom(\mathfrak{B}_{2})$. We notice that:
$$\mathfrak{B}_{1}(\phi,\psi)=\langle \mathfrak{L}_{2}^{-1}\tilde\phi,\tilde\psi\rangle, \quad \mathfrak{B}_{2}(\phi,\psi)=\langle \mathfrak{L}_{1}^{-1}\tilde\phi,\tilde\psi\rangle$$
and:
$$\mathfrak{Q}_{1}(\psi)=\langle\tilde\psi,\mathfrak{L}_{1}^{-1}\tilde\psi\rangle,\quad \mathfrak{Q}_{2}(\phi)=\langle\tilde\phi, \mathfrak{L}_{2}^{-1}\tilde\phi\rangle.$$
We infer that:
$$\left|\langle (\mathfrak{L}_{1}^{-1}- \mathfrak{L}_{2}^{-1})\tilde\phi,\tilde\psi\rangle\right|\leq\eta \|\mathfrak{L}_{1}^{-1}\|\|\mathfrak{L}_{2}^{-1}\| \|\tilde\phi\|\|\tilde\psi\|$$
and the result elementarily follows.
\end{proof}

\newpage
\section{Proofs in two dimensions}\label{2D}

\subsection{Proof of Theorem \ref{Thm-2D}}
Let us consider $\delta\leq 1$ and $K\geq 2\sup\frac{\kappa^2}{4}$.
\paragraph{A first approximation}
We let:
$$\mathcal{L}^{[2]}_{\eps,\delta}=\mathcal{L}^{[2]}_{\eps,\eps^{-\delta}\mathcal{A}_{\eps}}-\eps^{-2}\la_{1}^\Dir(\omega)+K$$
and
$$\mathcal{L}^{\app,[2]}_{\eps,\delta}=(i\dr_{s}+\eps^{1-\delta}\B(s,0)\tau)^2-\frac{\kappa^2}{4}-\eps^{-2}\dr_{\tau}^2-\eps^{-2}\la_{1}^\Dir(\omega)+K.$$
The corresponding quadratic forms, defined on $H^1_{0}(\Omega)$, are denoted by $\mathcal{Q}^{[2]}_{\eps,\delta}$ and $\mathcal{Q}^{[2]}_{\eps,\delta}$ whereas the sesquilinear forms are denoted by $\mathcal{B}^{[2]}_{\eps,\delta}$ and $\mathcal{B}^{[2]}_{\eps,\delta}$.
We can notice that:
$$\left|V_{\eps}(s,\tau)-\left(-\frac{\kappa(s)^2}{4}\right)\right|\leq C\eps$$
so that the operators $\mathcal{L}^{[2]}_{\eps,\delta}$ and $\mathcal{L}^{\app,[2]}_{\eps,\delta}$ are invertible for $\eps$ small enough. Moreover there exists $c>0$ such that for all $\varphi\in H^1_{0}(\Omega)$:
$$\mathcal{Q}^{[2]}_{\eps,\delta}(\varphi)\geq c\|\varphi\|^2,\quad \mathcal{Q}^{\app,[2]}_{\eps,\delta}(\varphi)\geq c\|\varphi\|^2.$$
Let $\phi,\psi\in H^1_{0}(\Omega)$.
We have to analyse the difference of the sesquilinear forms:
$$\mathcal{B}^{[2]}_{\eps,\delta}(\phi,\psi)-\mathcal{B}^{\app,[2]}_{\eps,\delta}(\phi,\psi).$$
We easily get:
$$\left|\langle V_{\eps}\phi,\psi\rangle-\langle -\frac{\kappa^2}{4}\phi,\psi\rangle\right|\leq C\eps\|\phi\|\|\psi\|\leq \tilde C\eps\sqrt{\mathcal{Q}^{[2]}_{\eps,\delta}(\psi)}\sqrt{\mathcal{Q}^{\app,[2]}_{\eps,\delta}(\phi)} .$$
We must investigate:
$$\langle m_{\eps}^2 (i\dr_{s}+b\mathcal{A}_{1}(s,\eps\tau))m_{\eps}\phi,(i\dr_{s}+b\mathcal{A}_{1}(s,\eps\tau))m_{\eps}\psi \rangle.$$
We notice that:
$$|\dr_{s}m_{\eps}|\leq C\eps,\quad |m_{\eps}-1|\leq C\eps.$$
We have:
\begin{multline*}
|\langle m_{\eps}^2 (i\dr_{s}+b\mathcal{A}_{1}(s,\eps\tau))m_{\eps}\phi,(i\dr_{s}+b\mathcal{A}_{1}(s,\eps\tau))(m_{\eps}-1)\psi \rangle|\\
\leq C\eps \|m_{\eps} (i\dr_{s}+b\mathcal{A}_{1}(s,\eps\tau))m_{\eps}\phi\|(\|\psi\|+\| m_{\eps}(i\dr_{s}+b\mathcal{A}_{1}(s,\eps\tau))\psi\|)\\
\leq C\eps  (\|(i\dr_{s}+b\mathcal{A}_{1}(s,\eps\tau))\phi\|+\|\phi\|)(\|\psi\|+\| m_{\eps}(i\dr_{s}+b\mathcal{A}_{1}(s,\eps\tau))\psi\|).
\end{multline*}
By the Taylor formula, we get (since $\delta\leq 1$):
\begin{equation}\label{Taylor-A1}
|\mathcal{A}_{1}(s,\eps\tau)-\eps b\B(s,0)\tau|\leq Cb\eps^2\leq C\eps.
\end{equation}
so that:
$$\|(i\dr_{s}+b\mathcal{A}_{1}(s,\eps\tau))\phi\|\leq \|(i\dr_{s}+\eps b\B(s,0)\tau)\phi\|+Cb\eps^2\|\phi\|.$$
We infer that:
\begin{multline*}
|\langle m_{\eps}^2 (i\dr_{s}+b\mathcal{A}_{1}(s,\eps\tau))m_{\eps}\phi,(i\dr_{s}+b\mathcal{A}_{1}(s,\eps\tau))(m_{\eps}-1)\psi \rangle|\\
\leq C\eps\left(\|\phi\|\|\psi\|+\|\phi\|\sqrt{\mathcal{Q}^{[2]}_{\eps,\delta}(\psi)}+\|\psi\|\sqrt{\mathcal{Q}^{\app,[2]}_{\eps,\delta}(\phi)}+\sqrt{\mathcal{Q}^{[2]}_{\eps,\delta}(\psi)}\sqrt{\mathcal{Q}^{\app,[2]}_{\eps,\delta}(\phi)}\right)\\
\leq \tilde C\eps\sqrt{\mathcal{Q}^{[2]}_{\eps,\delta}(\psi)}\sqrt{\mathcal{Q}^{\app,[2]}_{\eps,\delta}(\phi)}.
\end{multline*}
It remains to analyse:
$$\langle m_{\eps}^2 (i\dr_{s}+b\mathcal{A}_{1}(s,\eps\tau))m_{\eps}\phi,(i\dr_{s}+b\mathcal{A}_{1}(s,\eps\tau))\psi \rangle.$$
With the same kind of arguments, we deduce:
\begin{multline*}
|\langle m_{\eps}^2 (i\dr_{s}+b\mathcal{A}_{1}(s,\eps\tau))m_{\eps}\phi,(i\dr_{s}+b\mathcal{A}_{1}(s,\eps\tau))\psi \rangle-\langle (i\dr_{s}+b\mathcal{A}_{1}(s,\eps\tau))\phi,(i\dr_{s}+b\mathcal{A}_{1}(s,\eps\tau))\psi \rangle|\\
\leq  \tilde C\eps\sqrt{\mathcal{Q}^{[2]}_{\eps,\delta}(\psi)}\sqrt{\mathcal{Q}^{\app,[2]}_{\eps,\delta}(\phi)}.
\end{multline*}
We again use \eqref{Taylor-A1} to infer:
\begin{multline*}
\langle (i\dr_{s}+b\mathcal{A}_{1}(s,\eps\tau))\phi,(i\dr_{s}+b\mathcal{A}_{1}(s,\eps\tau))\psi \rangle-\langle (i\dr_{s}+b\mathcal{A}_{1}(s,\eps\tau))\phi,(i\dr_{s}+b\eps\B(s,0)\tau)\psi \rangle|\\
\leq C\eps\|(i\dr_{s}+b\mathcal{A}_{1}(s,\eps\tau))\phi\|\|\psi\|.
\leq \tilde C\eps \sqrt{\mathcal{Q}^{[2]}_{\eps,\delta}(\psi)}\sqrt{\mathcal{Q}^{\app,[2]}_{\eps,\delta}(\phi)}.
\end{multline*}
In the same way, we deduce:
\begin{multline*}
\langle (i\dr_{s}+b\mathcal{A}_{1}(s,\eps\tau))\phi,(i\dr_{s}+b\mathcal{A}_{1}(s,\eps\tau))\psi \rangle-\langle (i\dr_{s}+b\eps\B(s,0)\tau)\phi,(i\dr_{s}+b\eps\B(s,0)\tau)\psi \rangle|\\
\leq \tilde C\eps \sqrt{\mathcal{Q}^{[2]}_{\eps,\delta}(\psi)}\sqrt{\mathcal{Q}^{\app,[2]}_{\eps,\delta}(\phi)}.
\end{multline*}
We get:
$$\left|\mathcal{B}^{[2]}_{\eps,\delta}(\phi,\psi)-\mathcal{B}^{\app,[2]}_{\eps,\delta}(\phi,\psi)\right|\leq C\eps \sqrt{\mathcal{Q}^{[2]}_{\eps,\delta}(\psi)}\sqrt{\mathcal{Q}^{\app,[2]}_{\eps,\delta}(\phi)}.$$
By Lemma \ref{NRC}, we infer that:
\begin{equation}\label{first-app}
\left\|\left(\mathcal{L}^{[2]}_{\eps,\delta}\right)^{-1}-\left(\mathcal{L}^{\app,[2]}_{\eps,\delta}\right)^{-1}\right\|\leq \tilde C\eps.
\end{equation}
\paragraph{Case $\delta<1$.} The same kind of arguments provides:
$$\left|\mathcal{B}^{\app,[2]}_{\eps,\delta}(\phi,\psi)-\mathcal{B}^{\eff,[2]}_{\eps,\delta}(\phi,\psi)\right|\leq C\eps^{1-\delta} \sqrt{\mathcal{Q}^{\app,[2]}_{\eps,\delta}(\psi)}\sqrt{\mathcal{Q}^{\eff,[2]}_{\eps,\delta}(\phi)}$$
By Lemma \ref{NRC}, we get that:
$$\left\|\left(\mathcal{L}^{\app,[2]}_{\eps,\delta}\right)^{-1}-\left(\mathcal{L}^{\eff,[2]}_{\eps,\delta}\right)^{-1}\right\|\leq \tilde C\eps^{1-\delta}.$$
\paragraph{Case $\delta=1$.}
This case is slightly more complicated to analyse. We must estimate the difference the sesquilinear forms:
$$\mathcal{D}_{\eps}(\phi,\psi)=\mathcal{B}^{\app,[2]}_{\eps,1}(\phi,\psi)-\mathcal{B}^{\eff, [2]}_{\eps,1}(\phi,\psi).$$
We have:
$$\mathcal{D}_{\eps}(\phi,\psi)=\langle i\dr_{s}\phi, \B(s,0)\tau\psi \rangle+\langle \B(s,0)\tau\phi,i\dr_{s}\psi \rangle+\langle \B(s,0)^2\tau^2\phi,\psi\rangle-\|\tau J_{1}\|_{\omega}^2\langle \B(s,0)^2\phi,\psi\rangle.$$
We introduce the projection defined for $\varphi\in H^1_{0}(\Omega)$:
$$\Pi_{0}\varphi=\langle\varphi,J_{1} \rangle_{\omega}\, J_{1} $$
and we let, for all $\varphi\in H^1_{0}(\Omega)$:
$$\varphi^{\parallel}=\Pi_{0}\varphi,\quad\varphi^{\perp}=(\Id-\Pi_{0})\varphi.$$
We can write:
$$\mathcal{D}_{\eps}(\phi,\psi)=\mathcal{D}_{\eps}(\phi^\para,\psi^\para)+\mathcal{D}_{\eps}(\phi^\para,\psi^\perp)+\mathcal{D}_{\eps}(\phi^\perp,\psi^\para)+\mathcal{D}_{\eps}(\phi^\perp,\psi^\perp).$$
By using that $\langle \tau J_{1}, J_{1}\rangle_{\omega}=0$, we get:
$$\mathcal{D}_{\eps}(\phi^\para,\psi^\para)=0.$$
Then we have:
\begin{equation}\label{phi-pa-psi-pe1}
\|\tau J_{1}\|_{\omega}^2\langle \B(s,0)^2\phi^\para,\psi^\perp\rangle=0,\quad |\langle \B(s,0)^2\tau^2\phi^\para,\psi^\perp\rangle|\leq C\|\phi^\para\|\|\psi^\perp\|.
\end{equation}
Thanks to the min-max principle, we deduce:
\begin{equation}\label{gap}
\mathcal{Q}^{\app,[2]}_{\eps,1}(\psi^\perp)\geq\frac{\la_{2}^\Dir(\omega)-\la_{1}^\Dir(\omega)}{\eps^2}\|\psi^\perp\|^2,\quad \mathcal{Q}^{\eff,[2]}_{\eps,1}(\phi^\perp)\geq\frac{\la_{2}^\Dir(\omega)-\la_{1}^\Dir(\omega)}{\eps^2}\|\phi^\perp\|^2.
\end{equation}
Therefore we get:
$$|\langle \B(s,0)^2\tau^2\phi^\para,\psi^\perp\rangle|\leq C\eps\|\phi\|\sqrt{\mathcal{Q}^{\app, [2]}_{\eps,1}(\psi^\perp)}.$$
We have:
$$\mathcal{Q}^{\app,[2]}_{\eps,1}(\psi)=\mathcal{Q}^{\app, [2]}_{\eps,1}(\psi^\para)+\mathcal{Q}^{\app,[2]}_{\eps,1}(\psi^\perp)+\mathcal{B}^{\app, [2]}_{\eps,1}(\psi^\para,\psi^\perp)+\mathcal{B}^{\app,[2]}_{\eps,1}(\psi^\perp,\psi^\para).$$
We can write:
$$\mathcal{B}^{\app, [2]}_{\eps,1}(\psi^\para,\psi^\perp)=\langle (i\dr_{s}+\B(s,0)\tau)\psi^\para,(i\dr_{s}+\B(s,0)\tau)\psi^\perp\rangle.$$
We notice that:
\begin{equation}\label{phi-pa-psi-pe2}
\langle (i\dr_{s})\psi^\para,(i\dr_{s})\psi^\perp\rangle=0,\quad |\langle\B(s,0)\tau\psi^\para,\B(s,0)\tau\psi^\perp\rangle|\leq C\|\psi^\para\|\|\psi^\perp\|\leq C\|\psi\|^2.
\end{equation}
Moreover we have:
$$|\langle (i\dr_{s})\psi^\para, \B(s,0)\tau\psi^\perp\rangle|\leq C\|(i\dr_{s}\psi)^\para\|\|\psi^\perp\|\leq C\|i\dr_{s}\psi\|\|\psi\|\leq \tilde C\|\psi\|^2+\tilde C\|\psi\|\sqrt{\mathcal{Q}^{\app,[2]}_{\eps,1}(\psi)}.$$
The term $\mathcal{B}^{\app,[2]}_{\eps,1}(\psi^\perp,\psi^\para)$ can be analysed in the same way so that:
$$\mathcal{Q}^{\app,[2]}_{\eps,1}(\psi^\perp)\leq \mathcal{Q}^{\app,[2]}_{\eps,1}(\psi)+C\|\psi\|^2+C\|\psi\|\sqrt{\mathcal{Q}^{\app,[2]}_{\eps,1}(\psi)}\leq \tilde C(\|\psi\|^2+\mathcal{Q}^{\app,[2]}_{\eps,1}(\psi)).$$
We infer:
\begin{equation}\label{phi-pa-psi-pe3}
|\langle \B(s,0)^2\tau^2\phi^\para,\psi^\perp\rangle|\leq C\eps\|\phi\|\left(\|\psi\|+\sqrt{\mathcal{Q}^{\app,[2]}_{\eps,1}(\psi)}\right).
\end{equation}
We must now deal with the term
$$\langle i\dr_{s}\phi^\para, \B(s,0)\tau\psi^\perp \rangle.$$
We have:
$$|\langle i\dr_{s}\phi^\para, \B(s,0)\tau\psi^\perp \rangle|\leq C\|i\dr_{s}\phi\|\|\psi^\perp\|$$
and we easily deduce that:
\begin{equation}\label{phi-pa-psi-pe4}
|\langle i\dr_{s}\phi^\para, \B(s,0)\tau\psi^\perp \rangle|\leq C\eps\sqrt{\mathcal{Q}^{\eff,[2]}_{\eps,1}(\phi)}\left(\|\psi\|+\sqrt{\mathcal{Q}^{\app,[2]}_{\eps,1}(\psi)}\right),
\end{equation}
We also get the same kind of estimate by exchanging $\psi$ and $\phi$. Gathering \eqref{phi-pa-psi-pe1}, \eqref{phi-pa-psi-pe2}, \eqref{phi-pa-psi-pe3} and \eqref{phi-pa-psi-pe4}, we get the estimate:
$$|\mathcal{D}_{\eps}(\phi^\para,\psi^\perp)|\leq C\eps\sqrt{\mathcal{Q}^{\app,[2]}_{\eps,1}(\psi)}\sqrt{\mathcal{Q}^{\eff,[2]}_{\eps,1}(\phi)}.$$
By exchanging the roles of $\psi$ and $\phi$, we can also prove:
$$|\mathcal{D}_{\eps}(\phi^\perp,\psi^\para)|\leq C\eps\sqrt{\mathcal{Q}^{\app,[2]}_{\eps,1}(\psi)}\sqrt{\mathcal{Q}^{\eff,[2]}_{\eps,1}(\phi)}.$$
We must estimate $\mathcal{D}_{\eps}(\phi^\perp,\psi^\perp)$. With \eqref{gap}, we immediately deduce that:
$$|\langle \B(s,0)^2\tau^2\phi^\perp,\psi^\perp\rangle-\|\tau J_{1}\|_{\omega}^2\langle \B(s,0)^2\phi^\perp,\psi^\perp\rangle|\leq C\eps^2 \|\phi\|\|\psi\|.$$
We find that:
$$|\langle i\dr_{s}\phi^\perp, \B(s,0)\tau\psi^\perp \rangle|\leq C\|\psi^\perp\|\|i\dr_{s}\phi\|$$
and this term can treated as the others.
Finally we deduce the estimate:
$$|\mathcal{D}_{\eps}(\phi,\psi)|\leq C\eps \sqrt{\mathcal{Q}^{\app,[2]}_{\eps,1}(\psi)}\sqrt{\mathcal{Q}^{\eff,[2]}_{\eps,1}(\phi)}.$$
We apply Lemma~\ref{NRC} and the estimate \eqref{first-app} 
to obtain Theorem~\ref{Thm-2D}.

\subsection{Proof of Corollary~\ref{expansion-eigenvalues-2D}}
Let us expand the operator $\mathcal{L}^{[2]}_{\eps,b\mathcal{A}_{\eps}}$ in formal power series:
$$\mathcal{L}^{[2]}_{\eps,b\mathcal{A}_{\eps}}\sim \sum_{j=0} \eps^{j-2} L_{j},$$
where 
$$L_{0}=-\dr_{\tau}^2,\quad L_{1}=0,\quad L_{2}=(i\dr_{s}+\tau\B(s,0))^2-\frac{\kappa(s)^2}{4}.$$
We look for a quasimode in the form of a formal power series:
$$\psi\sim\sum_{j\geq 0}\eps ^j \psi_{j}$$
and a quasi-eigenvalue:
$$\gamma\sim\sum_{j\geq 0}\gamma_{j}\eps^{j-2}.$$
We must solve:
$$(L_{0}-\gamma_{0}) u_{0}=0.$$
We choose $\gamma_{0}=\frac{\pi^2}{4}$ and we take:
$$\psi_{0}(s,t)=f_{0}(s)J_{1}(\tau),$$
with $J_{1}(\tau)=\cos\left(\frac{\pi\tau}{2}\right)$.
Then, we must solve:
$$(L_{0}-\gamma_{0}) \psi_{1}=\gamma_{1}\psi_{0}.$$
We have $\gamma_{1}=0$ and $\psi_{1}=f_{1}(s)J_{1}(\tau).$
Then, we solve:
\begin{equation}\label{Eq-2}
(L_{0}-\gamma_{0}) \psi_{2}=\gamma_{2}u_{0}-L_{2}u_{0}.
\end{equation}
The Fredholm condition implies the equation:
$$-\dr_{s}^2 f+\left(\left(\frac{1}{3}+\frac{2}{\pi^2}\right)\B(s,0)^2-\frac{\kappa(s)^2}{4}\right)f_{0}=\mathcal{T}^{[2]}f_{0}=\gamma_{2}f_{0}$$
and we take for $\gamma_{2}=\gamma_{2,n}=\mu_{n}$ a negative eigenvalue of $\mathcal{T}^{[2]}$ and for $f_{0}$ a corresponding normalized eigenfunction (which has an exponential decay).

This leads to the choice:
$$\psi_{2}=\psi_{2}^\perp(s,\tau)+f_{2}(s)J_{1}(\tau),$$
where $\psi_{2}^\perp$ is the unique solution of (\ref{Eq-2}) which satisfies $\langle \psi_{2}^\perp, J_{1}\rangle_{\tau}=0.$
We can continue the construction at any order (see \cite{BDPR11, DomRay12} where this formal series method is used in a semiclassical context). We write $(\gamma_{j,n}, \psi_{j,n})$ instead of $(\gamma_{j}, \psi_{j})$ to emphasize the dependence on $n$ (determined in the choice of $\gamma_{2}$). We let:
\begin{equation}\label{quasi}
\Psi_{J,n}(\eps)=\sum_{j=0}^J \eps^{j}\psi_{j,n}, \mbox{ and }\Gamma_{J,n}(\eps)=\sum_{j=0}^J \eps^{-2+j} \gamma_{j,n}.
\end{equation}
A computation provides:
$$\|(\mathcal{L}^{[2]}_{\eps,b\mathcal{A}_{\eps}}-\Gamma_{J,n}(\eps))\Psi_{J,n}(\eps)\|\leq C\eps^{J+1}.$$
The spectral theorem implies that:
$$\dist(\Gamma_{J,n}(\eps),\sigma_{\mathsf{dis}}(\mathcal{L}^{[2]}_{\eps,b\mathcal{A}_{\eps}}))\leq C\eps^{J+1}.$$
It remains to use the spectral gap given by the approximation of the resolvent in Theorem~\ref{Thm-2D} and Corollary~\ref{expansion-eigenvalues-2D} follows.

\section{Proofs in three dimensions}\label{3D}

\subsection{Preliminaries}
We will adopt the following notation:
\begin{notation}\label{notation}
Given an open set $U \subset \R^d$ 
and a vector field $\F=\F(y_{1},\cdots, y_{d}) : U\to\R^d$ 
in dimension $d=2,3$, 
we will use in our computations the following notation:
$$
  \curl\F = 
\begin{cases}
  \dr_{y_{1}}\F_{2}-\dr_{y_{2}}\F_{1}
  &\mbox{if}\quad d=2,
  \\
  (\dr_{y_{2}}\F_{3}-\dr_{y_{3}}\F_{2},
  \dr_{y_{3}}\F_{1}-\dr_{y_{1}}\F_{3},
  \dr_{y_{1}}\F_{2}-\dr_{y_{2}}\F_{1})
  &\mbox{if}\quad d=3.
\end{cases}
$$	
The reader is warned that, 
if $(y_{1},\cdots, y_{d})$ represent curvilinear coordinates,
the outcome will differ from the usual (invariant) definition of~$\curl$.
\end{notation}
We recall the relations between $\mathcal{A}$, $\mathcal{B}$ and $\A$, $\B$. This can be done in terms of differential forms.
Let us consider the $1$-form:
$$\xi_{\A}=\A_{1}\dx x_{1}+\A_{2}\dx x_{2}+\A_{3}\dx x_{3}.$$
We consider $\Phi$ the diffeomorphism defined in \eqref{Phi}. 
The pull-back of $\xi_{\A}$ by $\Phi$ is given by:
$$\Phi^*\xi_{\A}=\mathcal{A}_{1}\dx t_{1}+\mathcal{A}_{2}\dx t_{2}+\mathcal{A}_{3}\dx t_{3}.$$
where $\mathcal{A}={}^t D\Phi \A(\Phi)$ since we have $x=\Phi(t)$ and we can write:
\begin{equation}\label{dx->dt}
\dx x_{i}=\sum_{j=1}^3\dr_{j}x_{i} \dx t_{j}.
\end{equation}
We can compute the exterior derivatives:
$$\dx\xi_{\A}=\B_{23}\dx x_{2}\wedge \dx x_{3}+\B_{13}\dx x_{1}\wedge \dx x_{3}+\B_{12}\dx x_{1}\wedge \dx x_{2}$$
and
$$d(\Phi^* \xi_{\A})=\mathcal{B}_{23}\dx t_{2}\wedge \dx t_{3}+\mathcal{B}_{13}\dx t_{1}\wedge \dx t_{3}+\mathcal{B}_{12}\dx t_{1}\wedge \dx t_{2},$$
with $\mathcal{B}=\curl\mathcal{A}$ and $\B=\curl\A$ (see Notation \ref{notation}). It remains to notice that the pull-back and the exterior derivative commute to get:
$$\Phi^* d\xi_{\A}= d(\Phi^*\xi_{\A})$$
and, using again \eqref{dx->dt}, it provides the relation:
$$\mathcal{B}={}^t Com(D\Phi)\B=\det (D\Phi) (D\Phi)^{-1}\B,$$
where ${}^t Com(D\Phi)$ denotes the transpose of the comatrix of $D\Phi$.
Let us give an interpretation of the components of $\mathcal{B}$.
A straightforward computation provides the following expression for $D\Phi$:
\begin{multline*}
[hT(s)+h_{2}(\sin\theta M_{2}-\cos\theta M_{3})+h_{3}(-\cos\theta M_{2}-\sin\theta M_{3}), \cos\theta M_{2}+\sin\theta M_{3}, -\sin\theta M_{2}+\cos\theta M_{3}]
\end{multline*}
so that $\det D\Phi=h$ and
$$\mathcal{B}_{23}=h(h^2+h_{2}^2+h_{3}^2)^{-1/2}\B\cdot T(s),\, \mathcal{B}_{13}=-h\B\cdot(-\cos\theta M_{2}-\sin\theta M_{3}),\, \mathcal{B}_{12}=h\B\cdot (-\sin\theta M_{2}+\cos\theta M_{3}).$$
Let us check that $\mathfrak{L}^{[3]}_{\eps,b\A}$ (whose quadratic form is denoted by $\mathfrak{Q}^{[3]}_{\eps,b\A}$) is unitarily equivalent to $\mathfrak{L}^{[3]}_{\eps,b\mathcal{A}}$ given in \eqref{L-frak3}. For that purpose we let:
$$G={}^t D\Phi D\Phi$$
and a computation provides:
$$G=\begin{pmatrix}
h^2+h_{2}^2+h_{3}^2&-h_{3} &-h_{2} \\
-h_{3}& 1& 0\\
-h_{2}& 0&1
\end{pmatrix}$$
and:
$$G^{-1}=\begin{pmatrix}
0&0&0 \\
0& 1& 0\\
0& 0&1
\end{pmatrix}+h^{-2}\begin{pmatrix} 1\\ h_{3}\\ h_{2} \end{pmatrix}\begin{pmatrix} 1&h_{3} &h_{2}\end{pmatrix}.
$$
We notice that $|G|=h^2$. In terms of quadratic form we write:
$$\mathfrak{Q}^{[3]}_{\eps,b\A}(\psi)=\int_{\R\times(\eps\omega)}|{}^t D\Phi^{-1}(-i\nabla_{t}+{}^t D\Phi \A(\Phi))|^2 \,h\dx t$$
and
\begin{eqnarray*}
&\mathfrak{Q}^{[3]}_{\eps,b\A}(\psi)=&\int_{\R\times(\eps\omega)} \left(|(-i\dr_{t_{2}}+b\mathcal{A}_{2})\psi|^2+ |(-i\dr_{t_{3}}+b\mathcal{A}_{3})\psi|^2\right)\,h \dx t\\
&                                                               &+\int_{\R\times(\eps\omega)}h^{-2}|\left(-i\dr_{s}+b\mathcal{A}_{1}+h_{3}(-i\dr_{t_{2}}+b\mathcal{A}_{2})+h_{2}(-i\dr_{t_{3}}+b\mathcal{A}_{3})\right)\psi|^2\,h\dx t
\end{eqnarray*}
so that:
\begin{multline*}
\mathfrak{Q}^{[3]}_{\eps,b\A}(\psi)\\
=\int_{\R\times(\eps\omega)} \left(|(-i\dr_{t_{2}}+b\mathcal{A}_{2})\psi|^2+ |(-i\dr_{t_{3}}+b\mathcal{A}_{3})\psi|^2+h^{-2}|(-i\dr_{s}+b\mathcal{A}_{1}-i\theta'\dr_{\alpha}+\mathcal{R})\psi|^2\right)\,h\dx t.
\end{multline*}

\paragraph{Choice of gauge} Since $\omega$ is simply connected (and so is $\Omega_{\eps}$) we may change the gauge and assume that the vector potential is given by:
\begin{eqnarray}\label{explicit-A}
\nonumber&\mathcal{A}_{1}(s,t_{2},t_{3})&=-\frac{t_{2} t_{3}\dr_{s}\mathcal{B}_{23}(s,0,0)}{2}-\int_{0}^{t_{2}} \mathcal{B}_{12}(s,\tilde t_{2},t_{3})\dx \tilde t_{2}-\int_{0}^{t_{3}}\mathcal{B}_{13}(s,0,\tilde t_{3})\dx \tilde t_{3} ,\\
&\mathcal{A}_{2}(s,t_{2},t_{3})&=-\frac{ t_{3}\mathcal{B}_{23}(s,0,0)}{2},\\
\nonumber&\mathcal{A}_{3}(s,t_{2},t_{3})&=-\frac{ t_{2}\mathcal{B}_{23}(s,0,0)}{2}+\int_{0}^{t_{2}}\mathcal{B}_{23}(s,\tilde t_{2},t_{3})\dx \tilde t_{2}.
\end{eqnarray}
In other words, thanks to the Poincar\'e lemma, there exists a (smooth) phase function $\rho$ such that $D\Phi \A(\Phi)+\nabla_{t}\rho=\mathcal{A}$. In particular, we have: $\mathcal{A}_{j}(s,0)=0, \dr_{j}\mathcal{A}_{j}(s,0)=0$ for $j\in\{1,2,3\}$.

\subsection{Proof of Theorem \ref{Thm-3D}}
Let us consider $\delta\leq 1$ and $K\geq 2\sup\frac{\kappa^2}{4}$.
\paragraph{A first approximation}
We let:
$$\mathcal{L}^{[3]}_{\eps,\delta}=\mathcal{L}^{[3]}_{\eps,\eps^{-\delta}\mathcal{A}_{\eps}}-\eps^{-2}\la_{1}^\Dir(\omega)+K$$
and
$$\mathcal{L}^{\app,[3]}_{\eps,\delta}=\sum_{j=2,3} (-i\eps^{-1}\dr_{\tau_{j}}+b\mathcal{A}^\lin_{j,\eps})^2+(-i\dr_{s}+b\mathcal{A}_{1,\eps}^\lin-i\theta'\dr_{\alpha})^2-\frac{\kappa^2}{4}-\eps^{-2}\dr_{\tau}^2-\eps^{-2}\la_{1}^\Dir(\omega)+K,$$
where:
$$\mathcal{A}^\lin_{j,\eps}(s,\tau)=\mathcal{A}_{j}(s,0)+\eps\tau_{2}\dr_{2}\mathcal{A}_{j}(s,0)+\eps\tau_{3}\dr_{3}\mathcal{A}_{j}(s,0).$$
We recall that $\mathcal{A}$ is given by \eqref{explicit-A} and that $\mathcal{L}^{[3]}_{\eps,\eps^{-\delta}\mathcal{A}_{\eps}}$ is defined in \eqref{L3}. We have to analyse the difference of the corresponding sesquilinear forms:
$$\mathcal{B}^{[3]}_{\eps,\delta}(\phi,\psi)-\mathcal{B}^{\app, [3]}_{\eps,\delta}(\phi,\psi).$$
Let us deal with the term:
$$\langle h_{\eps}^{-1}(-i\dr_{s}-i\theta'\dr_{\alpha}+\mathcal{R}_{\eps})h_{\eps}^{-1/2}\phi,(-i\dr_{s}-i\theta'\dr_{\alpha}+\mathcal{R}_{\eps})h_{\eps}^{-1/2}\psi\rangle.$$
Since we have $|(-i\dr_{s}-i\theta'\dr_{\alpha}) h_{\eps}^{-1/2}|\leq C\eps$, we get:
\begin{multline*}
|\langle h_{\eps}^{-1}(-i\dr_{s}-i\theta'\dr_{\alpha}+\mathcal{R}_{\eps})h_{\eps}^{-1/2}\phi,(-i\dr_{s}-i\theta'\dr_{\alpha}+\mathcal{R}_{\eps})(h_{\eps}^{-1/2}-1)\psi\rangle|\\
\leq C\eps \|h_{\eps}^{-1/2}(-i\dr_{s}-i\theta'\dr_{\alpha}+\mathcal{R}_{\eps})h_{\eps}^{-1/2}\phi\|(\|\psi\|+\|(-i\dr_{s}-i\theta'\dr_{\alpha}+\mathcal{R}_{\eps})\psi\|)
\end{multline*}
and we get:
\begin{multline*}
|\langle h_{\eps}^{-1}(-i\dr_{s}-i\theta'\dr_{\alpha}+\mathcal{R}_{\eps})h_{\eps}^{-1/2}\phi,(-i\dr_{s}-i\theta'\dr_{\alpha}+\mathcal{R}_{\eps})(h_{\eps}^{-1/2}-1)\psi\rangle|\\
\leq C\eps\left(\|\phi\|\|\psi\|+\|\phi\|\sqrt{\mathcal{Q}^{[3]}_{\eps,\delta}(\psi)}+\|\psi\|\sqrt{\mathcal{Q}^{\app [3]}_{\eps,\delta}(\phi)}+\sqrt{\mathcal{Q}^{\app, [3]}_{\eps,\delta}(\phi)}\sqrt{\mathcal{Q}^{[3]}_{\eps,\delta}(\psi)}\right).
\end{multline*}
With the same kind of estimates, it follows that:
\begin{multline*}
|\langle h_{\eps}^{-1}(-i\dr_{s}-i\theta'\dr_{\alpha}+\mathcal{R}_{\eps})h_{\eps}^{-1/2}\phi,(-i\dr_{s}-i\theta'\dr_{\alpha}+\mathcal{R}_{\eps})h_{\eps}^{-1/2}\psi\rangle\\
-\langle (-i\dr_{s}-i\theta'\dr_{\alpha}+\mathcal{R}_{\eps})\phi,(-i\dr_{s}-i\theta'\dr_{\alpha}+\mathcal{R}_{\eps})\psi\rangle|\\
\leq C\eps\left(\|\phi\|\|\psi\|+\|\phi\|\sqrt{\mathcal{Q}^{[3]}_{\eps,\delta}(\psi)}+\|\psi\|\sqrt{\mathcal{Q}^{\app [3]}_{\eps,\delta}(\phi)}+\sqrt{\mathcal{Q}^{\app, [3]}_{\eps,\delta}(\phi)}\sqrt{\mathcal{Q}^{[3]}_{\eps,\delta}(\psi)}\right).
\end{multline*}
With the Taylor formula we notice that:
$$|\mathcal{A}_{j,\eps}(s,\tau)-\mathcal{A}^\lin_{j,\eps}(s,\tau)|\leq C\eps^2.$$
We notice that $|\mathcal{R}_{\eps}|\leq C\eps$ due to the properties of the vector potential $\mathcal{A}$ (see \eqref{explicit-A}). Then we can apply the same technique as in Section \ref{2D} to deduce:
$$|\mathcal{B}^{[3]}_{\eps,\delta}(\phi,\psi)-\mathcal{B}^{\app, [3]}_{\eps,\delta}(\phi,\psi)|\leq C\eps\sqrt{\mathcal{Q}^{\app, [3]}_{\eps,\delta}(\phi)}\sqrt{\mathcal{Q}^{[3]}_{\eps,\delta}(\psi)} .$$
and then:
\begin{equation}\label{first-approx}
\left\|(\mathcal{L}^{[3]}_{\eps,\delta})^{-1}-(\mathcal{L}^{\app,[3]}_{\eps,\delta})^{-1}\right\|\leq \tilde C\eps.
\end{equation}
\paragraph{Case $\delta<1$}
This case is similar to the case in dimension $2$ since $|b\mathcal{A}^\lin_{j,\eps}|\leq C\eps^{1-\delta}$. If we let:
$$\mathcal{L}^{\appp,[3]}_{\eps,\delta}=\sum_{j=2,3} (-i\eps^{-1}\dr_{\tau_{j}})^2+(-i\dr_{s}-i\theta'\dr_{\alpha})^2-\frac{\kappa^2}{4}-\eps^{-2}\dr_{\tau}^2-\eps^{-2}\la_{1}^\Dir(\omega)+K,$$
we easily get:
$$\left\|(\mathcal{L}^{\appp,[3]}_{\eps,\delta})^{-1}-(\mathcal{L}^{\app,[3]}_{\eps,\delta})^{-1}\right\|\leq \tilde C\eps^{1-\delta}.$$
It remains to decompose the sesquilinear form associated with $\mathcal{L}^{\appp,[3]}_{\eps,\delta}$ by using the orthogonal projection $\Pi_{0}$ and the analysis follows the same lines as in dimension $2$.
\paragraph{Case $\delta=1$} This case cannot be analysed in the same way as in dimension 2. Using the explicit expression of the vector potential \eqref{explicit-A}, we can write our approximated operator in the form:
\begin{eqnarray*}
\mathcal{L}^{\appp,[3]}_{\eps,1}=&\left(-\eps^{-1}i\dr_{\tau_{2}}-\frac{\mathcal{B}_{23}(s,0,0)}{2}\tau_{3}\right)^2+\left(-\eps^{-1}i\dr_{\tau_{3}}+\frac{\mathcal{B}_{23}(s,0,0)}{2}\tau_{2}\right)^2\\
                                                           &+(-i\dr_{s}-i\theta'\dr_{\alpha}-\tau_{2}\mathcal{B}_{12}(s,0,0)-\tau_{3}\mathcal{B}_{13}(s,0,0))^2-\eps^{-2}\la_{1}^\Dir(\omega)+K.
\end{eqnarray*}

\paragraph{Perturbation theory}
Let us introduce the operator on $L^2(\omega)$ (with Dirichlet boundary condition) and depending on $s$:
$$\mathcal{P}_{\eps}^2=\left(-\eps^{-1}i\dr_{\tau_{2}}-\frac{\mathcal{B}_{23}(s,0,0)}{2}\tau_{3}\right)^2+\left(-\eps^{-1}i\dr_{\tau_{3}}+\frac{\mathcal{B}_{23}(s,0,0)}{2}\tau_{2}\right)^2.$$
Thanks to perturbation theory the lowest eigenvalue $\nu_{1,\eps}(s)$ of $\mathcal{P}_{\eps}(s)$ is simple and we may consider an associated $L^2$ normalized eigenfunction $u_{\eps}(s)$. Let us provide a estimate for the eigenpair $(\nu_{1,\eps}(s),u_{\eps}(s))$. We have to be careful with the dependence on $s$ in the estimates. Firstly, we notice that there exist $\eps_{0}>0$ and $C>0$ such that for all $s$, $\eps\in(0,\eps_{0})$ and all $\psi\in H^1_{0}(\omega)$:
\begin{multline}
\int_{\omega} \left|\left(-\eps^{-1}i\dr_{\tau_{2}}-\frac{\mathcal{B}_{23}(s,0,0)}{2}\tau_{3}\right)\psi\right|^2+\left|\left(-\eps^{-1}i\dr_{\tau_{3}}+\frac{\mathcal{B}_{23}(s,0,0)}{2}\tau_{2}\right)\psi\right|^2 \dx \tau\\
\geq \eps^{-2} \int_{\omega} |\dr_{\tau_{2}}\psi|^2+|\dr_{\tau_{3}}\psi|^2 \dx\tau-C\eps^{-1}\|\psi\|^2.
\end{multline}
From the min-max principle we infer that:
\begin{equation}\label{sgap}
\nu_{n,\eps}(s)\geq \eps^{-2}\la^\Dir_{n}(\omega)-C\eps^{-1}.
\end{equation}
Let us analyse the corresponding upper bound. Thanks to the Fredholm alternative, we may introduce $R_{\omega}$ the unique function such that:
\begin{equation}\label{J1}
(-\Delta^\Dir_{\omega}-\la_{1}^{\Dir}(\omega)) R_{\omega}=D_{\alpha} J_{1},\quad \langle R_{\omega}, J_{1}\rangle_{\omega}=0.
\end{equation}
We use $v_{\eps}=J_{1}+\eps\mathcal{B}_{23}(s,0,0)R_{\omega}$ as test function for $\mathcal{P}_{\eps}^2$ and an easy computation provides that there exist $\eps_{0}>0$ and $C>0$ such that for all $s$, $\eps\in(0,\eps_{0})$:
$$\left\|\left(\mathcal{P}_{\eps}^2-\left(\eps^{-2}\la^\Dir_{1}(\omega)+\mathcal{B}_{23}^2(s,0,0)\left(\frac{\|\tau J_{1}\|_{\omega}^2}{4}-\langle D_{\alpha}R_{\omega},J_{1}\rangle_{\omega}\right)\right)\right)v_{\eps}\right\|_{\omega}\leq C\eps.$$
The spectral theorem implies that there exists $n(\eps,s)\geq 1$ such that:
$$\left|\nu_{n(\eps,s),\eps}(s)-\eps^{-2}\la^\Dir_{1}(\omega)-\mathcal{B}_{23}^2(s,0,0)\left(\frac{\|\tau J_{1}\|_{\omega}^2}{4}-\langle D_{\alpha}R_{\omega},J_{1}\rangle_{\omega}\right)\right|\leq C\eps.$$
Due to the spectral gap uniform in $s$ given by \eqref{sgap} we deduce that  there exist $\eps_{0}>0$ and $C>0$ such that for all $s$, $\eps\in(0,\eps_{0})$:
$$\left|\nu_{1,\eps}(s)-\eps^{-2}\la^\Dir_{1}(\omega)-\mathcal{B}_{23}^2(s,0,0)\left(\frac{\|\tau J_{1}\|^2}{4}-\langle D_{\alpha}R_{\omega},J_{1}\rangle_{\omega}\right)\right|\leq C\eps.$$
This new information provides:
$$\left\|\left(\mathcal{P}_{\eps}^2-\nu_{1,\eps}(s)\right)v_{\eps}\right\|_{\omega}\leq \tilde C\eps$$
and thus:
$$\left\|\left(\mathcal{P}_{\eps}^2-\nu_{1,\eps}(s)\right)(v_{\eps}-\langle v_{\eps},u_{\eps}\rangle_{\omega} u_{\eps})\right\|_{\omega}\leq \tilde C\eps.$$
so that, with the spectral theorem and the uniform gap between the eigenvalues:
$$\left\|v_{\eps}-\langle v_{\eps},u_{\eps}\rangle_{\omega} u_{\eps}\right\|_{\omega}\leq C\eps^3.$$
Up to changing $u_{\eps}$ in $-u_{\eps}$, we infer that :
$$||\langle v_{\eps},u_{\eps}\rangle_{\omega}|-\|v_{\eps}\|_{\omega}|\leq C\eps^3,\quad \left\|v_{\eps}-\|v_{\eps}\|_{\omega} u_{\eps}\right\|_{\omega}\leq \tilde C\eps^3.$$
Therefore we get:
$$\left\|u_{\eps}-\tilde v_{\eps}\right\|_{\omega}\leq C\eps^3,\quad \tilde v_{\eps}=\frac{v_{\eps}}{\|v_{\eps}\|_{\omega}}$$
and this is easy to deduce:
\begin{equation}\label{approx-H1}
\left\|\nabla_{\tau_{2},\tau_{3}}\left(u_{\eps}-\tilde v_{\eps}\right)\right\|_{\omega}\leq C\eps^3.
\end{equation}
\paragraph{Projection arguments} 
We shall analyse the difference of the sesquilinear forms:
$$\mathcal{D}_{\eps}(\phi,\psi)=\mathcal{L}^{\appp,[3]}_{\eps,1}(\phi,\psi)-\mathcal{L}^{\eff,[3]}_{\eps,1}(\phi,\psi).$$
We write:
$$\mathcal{D}_{\eps}(\phi,\psi)=\mathcal{D}_{\eps,1}(\phi,\psi)+\mathcal{D}_{\eps,2}(\phi,\psi),$$
where
$$\mathcal{D}_{\eps,1}(\phi,\psi)=\langle\mathcal{P}_{\eps}\phi,\mathcal{P}_{\eps}\psi\rangle-\left\langle-\eps^{-2}\Delta_{\omega}^\Dir+\mathcal{B}_{23}^2(s,0,0)\left(\frac{\|\tau J_{1}\|_{\omega}^2}{4}-\langle D_{\alpha}R_{\omega},J_{1}\rangle_{\omega}\right)\right\rangle$$
and
$$\mathcal{D}_{\eps,2}(\phi,\psi)=\langle\mathcal{M}\phi,\psi\rangle-\langle\mathcal{M}^\eff\phi,\psi\rangle,$$
with:
$$\mathcal{M}=\left(-i\dr_{s}-i\theta'\dr_{\alpha}-\tau_{2}\mathcal{B}_{12}(s,0,0)-\tau_{3}\mathcal{B}_{13}(s,0,0)\right)^2,$$
$$\mathcal{M}^\eff=\langle(-i\dr_{s}-i\theta'\dr_{\alpha}-\mathcal{B}_{12}(s,0,0)\tau_{2}-\mathcal{B}_{13}(s,0,0)\tau_{3})^2 \Id(s)\otimes J_{1}, \Id(s)\otimes J_{1}\rangle_{\omega}.$$

\paragraph{Estimate of $\mathcal{D}_{\eps,1}(\phi,\psi)$}
We introduce the projection on $u_{\eps}(s)$:
$$\Pi_{\eps,s}\varphi=\langle\varphi ,u_{\eps}\rangle_{\omega}\, u_{\eps}(s)$$
and, for $\varphi\in H^1_{0}(\Omega)$, we let:
$$\varphi^{\para_{\eps}}=\Pi_{\eps,s}\varphi,\quad \varphi^{\perp_{\eps}}=\varphi-\Pi_{\eps,s}\varphi.$$
We can write the formula:
$$\mathcal{D}_{\eps,1}(\phi,\psi)=\mathcal{D}_{\eps,1}(\phi^{\para_{\eps}},\psi^\para)+\mathcal{D}_{\eps,1}(\phi^{\para_{\eps}},\psi^\perp)+\mathcal{D}_{\eps,1}(\phi^{\perp_{\eps}},\psi^\para)+\mathcal{D}_{\eps,1}(\phi^{\perp_{\eps}},\psi^\perp),$$
where $\psi^\para=\Pi_{0}\psi=\langle\psi, J_{1}\rangle_{\omega}\, J_{1}$ and $\psi^\perp=\psi-\psi^\para$.
\begin{remark}
We notice that the decomposition of the sesquilinear form is performed with respect to the two projections $\Pi_{0}$ and $\Pi_{\eps,s}$. This is due to the fact that we need to catch the effect of the magnetic field in the subprincipal terms.
\end{remark}
Let us analyse $\mathcal{D}_{\eps,1}(\phi^{\para_{\eps}},\psi^\para)$. We have to estimate:
$$\left\langle\left(\frac{\mathcal{B}_{23}}{\eps} D_{\alpha}+\frac{\mathcal{B}_{23}^2}{4}(\tau_{2}^2+\tau_{3}^2)\right)\phi^{\para_{\eps}},\psi^\para\right\rangle-\left\langle\frac{\mathcal{B}_{23}^2(s,0,0)}{4}\left(\|\tau J_{1}\|_{\omega}^2-4\langle D_{\alpha}R_{\omega},J_{1}\rangle_{\omega}\right)\phi^{\para_{\eps}},\psi^\para\right\rangle$$
We notice that:
\begin{multline*}
\left|\left\langle\left(\frac{\mathcal{B}_{23}^2}{2}(\tau_{2}^2+\tau_{3}^2)\right)\phi^{\para_{\eps}},\psi^\para\right\rangle-\left\langle\frac{\mathcal{B}_{23}^2(s,0,0)}{4}\|\tau J_{1}\|_{\omega}^2 \phi^{\para_{\eps}},\psi^\para\right\rangle\right|\\
\leq C \left|\int_{\omega} (\tau^2 u_{\eps}J_{1}-\tau^2 J_{1}^2)\dx\tau\right| \|\langle \phi,u_{\eps}\rangle\|\|\langle \psi,J_{1}\rangle\|.
\end{multline*}
Thanks to the approximation result, we get (uniformly in $s$):
$$\left|\int_{\omega} (\tau^2 u_{\eps}J_{1}-\tau^2 J_{1}^2)\dx\tau\right|\leq C\eps$$
and thus:
$$\left|\left\langle\left(\frac{\mathcal{B}_{23}^2}{2}(\tau_{2}^2+\tau_{3}^2)\right)\phi^{\para_{\eps}},\psi^\para\right\rangle-\left\langle\frac{\mathcal{B}_{23}^2(s,0,0)}{4}\|\tau J_{1}\|_{\omega}^2 \phi^{\para_{\eps}},\psi^\para\right\rangle\right|\\
\leq C\eps\|\phi\|\|\psi\|.$$
Then, we get:
$$\left|\left\langle\frac{\mathcal{B}_{23}(s,0,0)}{\eps} D_{\alpha}\phi^{\para_{\eps}},\psi^\para\right\rangle-\left\langle \langle\psi, J_{1}\rangle_{\omega}\frac{\mathcal{B}_{23}(s,0,0)}{\eps} D_{\alpha}\tilde v_{\eps}  , \langle\phi, u_\eps\rangle_{\omega} J_{1} \right\rangle\right|\leq C\eps^2\|\phi\|\|\psi\|$$ 
and a computation gives:
$$\left\langle \langle\psi, J_{1}\rangle_{\omega}\frac{\mathcal{B}_{23}(s,0,0)}{\eps} D_{\alpha}v_{\eps}  , \langle\phi, u_{\eps}\rangle_{\omega} J_{1} \right\rangle=\left\langle D_{\alpha} R_{\omega} , J_{1}\rangle_{\omega}\langle\mathcal{B}_{23}(s,0,0)^2 \langle\psi, J_{1}\rangle_{\omega}, \langle\phi, u_{\eps}\rangle_{\omega}  \right\rangle$$
and in the same way we get:
$$\left|\langle\mathcal{B}_{23}(s, 0,0)^2 \langle\psi, J_{1}\rangle_{\omega}, \langle\phi, u_{\eps}\rangle_{\omega}\rangle-\langle \mathcal{B}_{23}(s,0,0)^2\psi^\para ,\phi^{\para_{\eps}} \rangle\right|\leq C\eps\|\phi\|\|\psi\|.$$
Therefore we deduce:
$$|\mathcal{D}_{\eps,1}(\phi^{\para_{\eps}},\psi^\para)|\leq C\eps\|\phi\|\|\psi\|\leq \tilde C\eps \sqrt{\mathcal{Q}^{\eff,[3]}_{\eps,1}(\phi)} \sqrt{\mathcal{Q}^{\app2,[3]}_{\eps,1}(\psi)}.$$
Let us now deal with $\mathcal{D}_{\eps,1}(\phi^{\para_{\eps}},\psi^\perp)$. We notice that:
$$\|\psi^\perp\|\leq C\eps\|\eps^{-1}\nabla\psi^\perp\|\leq C\eps\|\eps^{-1}\nabla\psi\|,$$
since $\langle\nabla\psi^\para, \nabla\psi^\perp\rangle=\langle -\Delta^\Dir\psi^\para,\psi^\perp\rangle=0$. In addition, we easily get:
$$\|\eps^{-1}\nabla\psi\|\leq C\sqrt{\mathcal{Q}^{\app2,[3]}_{\eps,1}(\psi)}+C\|\psi\|.$$
Therefore we deduce that:
$$\|\psi^\perp\|\leq \tilde C\eps \sqrt{\mathcal{Q}^{\app2,[3]}_{\eps,1}(\psi)}.$$
The most delicate term to analyse is:
\begin{multline*}
\left|\left\langle\frac{\mathcal{B}_{23}}{2}\eps^{-1}D_{\alpha}\phi^{\para_{\eps}},\psi^\perp\right\rangle\right|\leq C\|\psi^\perp\|\|\eps^{-1}D_{\alpha}\phi^{\para_{\eps}}\|\leq \tilde C\|\psi^\perp\|\|\eps^{-1}\nabla\phi^{\para_{\eps}}\|\\
\leq C\|\psi^\perp\|\|\mathcal{P}_{\eps}\phi^{\para_{\eps}}\|+C\|\psi^\perp\|\|\phi\|.
\end{multline*}
But we have:
$$\|\mathcal{P}_{\eps}\phi\|^2=\|\mathcal{P}_{\eps}\phi^{\para_{\eps}}\|^2+\|\mathcal{P}_{\eps}\phi^{\perp_{\eps}}\|^2+2\langle \mathcal{P}_{\eps}\phi^{\para_{\eps}},\mathcal{P}_{\eps}\phi^{\perp_{\eps}}\rangle=\|\mathcal{P}_{\eps}\phi^{\para_{\eps}}\|^2+\|\mathcal{P}_{\eps}\phi^{\perp_{\eps}}\|^2$$
and:
$$\|\mathcal{P}_{\eps}\phi\|^2\leq C(\|\eps^{-1}\nabla\phi\|^2+\|\phi\|^2)\leq C\mathcal{Q}^{\eff,[3]}_{\eps,1}(\phi)+C\|\phi\|^2$$
so that:
$$\left|\left\langle\frac{\mathcal{B}_{23}}{2}\eps^{-1}D_{\alpha}\phi^{\para_{\eps}},\psi^\perp\right\rangle\right|\leq C\eps \sqrt{\mathcal{Q}^{\app2,[3]}_{\eps,1}(\psi)}\sqrt{\mathcal{Q}^{\eff,[3]}_{\eps,1}(\phi)}.$$
The term $\mathcal{D}_{\eps,1}(\phi^{\perp_{\eps}},\psi^\para)$ can be analysed with the same arguments since we have:
$$C\eps^{-2}\|\phi^{\perp_{\eps}}\|^2\leq\|\mathcal{P}_{\eps}\phi^{\perp_{\eps}}\|^2\leq \|\mathcal{P}_{\eps}\phi\|^2\leq \tilde C\|\eps^{-1}\nabla\phi\|^2+\tilde C\|\phi\|^2\leq C\mathcal{Q}^{\eff,[3]}_{\eps,1}(\phi).$$
The investigation of $\mathcal{D}_{\eps,1}(\phi^{\perp_{\eps}},\psi^\perp)$ goes along the same lines.
Therefore we have proved that:
\begin{equation}\label{D1}
|\mathcal{D}_{\eps,1}(\phi,\psi)|\leq C\eps\sqrt{\mathcal{Q}^{\app2,[3]}_{\eps,1}(\psi)}\sqrt{\mathcal{Q}^{\eff,[3]}_{\eps,1}(\phi)}.
\end{equation}
\paragraph{Estimate of $\mathcal{D}_{\eps,2}(\phi,\psi)$} We use the decomposition of $\phi$ and $\psi$ with respect to $J_{1}$ and its orthogonal. We have:
$$\mathcal{D}_{2,\eps}(\phi^\para,\psi^\para)=0.$$
Let us explain how to deal with term $\mathcal{D}_{2,\eps}(\phi^\para,\psi^\perp)$. The worst term can be bounded by $\|\dr_{s}\phi^\para\|\|\psi^\perp\|$ and we have:
$$\|\dr_{s}\phi^\para\|\leq \|\dr_{s}\phi\|\leq C \mathcal{Q}^{\eff,[3]}_{\eps,1}(\phi)+C\|\phi\|\leq\tilde C\mathcal{Q}^{\eff,[3]}_{\eps,1}(\phi).$$
In addition we have:
$$C\eps^{-2}\|\psi^\perp\|^2\leq \|\eps^{-1}\nabla\psi^\perp\|^2\leq  \|\eps^{-1}\nabla\psi\|^2\leq C\|\mathcal{P}_{\eps}\psi\|^2+C\|\psi\|^2\leq\tilde C\mathcal{Q}^{\app2,[3]}_{\eps,1}(\psi).$$
We infer that:
$$\|\dr_{s}\phi^\para\|\|\psi^\perp\|\leq C\eps\sqrt{\mathcal{Q}^{\app2,[3]}_{\eps,1}(\psi)}\sqrt{\mathcal{Q}^{\eff,[3]}_{\eps,1}(\phi)}.$$
The analysis of $\mathcal{D}_{2,\eps}(\phi^\perp,\psi^\para)$ and $\mathcal{D}_{2,\eps}(\phi^\perp,\psi^\perp)$ can be performed in the same way and we get
\begin{equation}\label{D2}
|\mathcal{D}_{\eps,2}(\phi,\psi)|\leq C\eps\sqrt{\mathcal{Q}^{\app2,[3]}_{\eps,1}(\psi)}\sqrt{\mathcal{Q}^{\eff,[3]}_{\eps,1}(\phi)}.
\end{equation}
Combining \eqref{D1} and \eqref{D2}, we infer that:
$$|\mathcal{D}_{\eps}(\phi,\psi)|\leq C\eps\sqrt{\mathcal{Q}^{\app2,[3]}_{\eps,1}(\psi)}\sqrt{\mathcal{Q}^{\eff,[3]}_{\eps,1}(\phi)}.$$
With Lemma \ref{NRC} we infer:
\begin{equation}\label{second-approx}
\left\| \left(\mathcal{L}^{\appp,[3]}_{\eps,1}\right)^{-1}-\left(\mathcal{L}^{\eff,[3]}_{\eps,1}\right)^{-1}\right\|\leq C\eps.
\end{equation}
Finally we deduce Theorem \ref{Thm-3D} from \eqref{first-approx} and \eqref{second-approx}.

\subsection{Proof of Corollary~\ref{expansion-eigenvalues-3D}}
For the asymptotic expansions of the eigenvalues claimed in Corollary~\ref{expansion-eigenvalues-3D}, we leave the proof to the reader since it is a slight adaptation of the proof of Corollary~\ref{expansion-eigenvalues-2D}.

\section{Repulsive effect of magnetic fields}\label{rep}

\subsection{Proof of Theorem~\ref{stability}}
Let us perform a preliminary computation.
\begin{lemma}
We have, for all $R>0$ and $\phi\in\mathcal{C}_{0}^{\infty}(\Omega)$:
\begin{equation}\label{IMS-J0}
\int_{\Omega(R)} |(-i\nabla+\A)J_{1}\phi|^2\dx t=\int_{\Omega(R)} J_{1}^2|(-i\nabla+\A)\phi|^2\dx t+\la_{1}^\Dir(\omega)\int_{\Omega(R)} J_{1}^2|\phi|^2\dx t.
\end{equation}
\end{lemma}
\begin{proof}
We have, for all $R>0$ and $\phi\in\mathcal{C}_{0}^{\infty}(\Omega)$:
$$\int_{\Omega(R)} |(-i\nabla+\A)J_{1}\phi|^2\dx t=\int_{\Omega(R)} |-i\phi\nabla J_{1}+J_{1}(-i\nabla+\A)\phi|^2\,\dx t.$$
This becomes:
\begin{eqnarray*}
&&\int_{\Omega(R)} |(-i\nabla+\A)J_{1}\phi|^2\dx t\\
&&=\int_{\Omega(R)} |\phi|^2 |\nabla J_{1}|^2\dx t+\int_{\Omega(R)} J_{1}^2|(-i\nabla+\A)\phi|^2\dx t+2\Re\left(\int_{\Omega(R)} \overline{\phi} J_{1}\nabla J_{1}\cdot \nabla \phi\dx t\right)\\
&&=\int_{\Omega(R)} |\phi|^2 |\nabla J_{1}|^2\dx t+\int_{\Omega(R)} J_{1}^2|(-i\nabla+\A)\phi|^2\dx t+\int_{\Omega(R)} J_{1}\nabla J_{1}\cdot\nabla\left(|\phi|^2\right) \dx t\\
&&=\int_{\Omega(R)} |\phi|^2 |\nabla J_{1}|^2\dx t+\int_{\Omega(R)} J_{1}^2|(-i\nabla+\A)\phi|^2\dx t-\int_{\Omega(R)} \nabla\cdot(J_{1}\nabla J_{1})|\phi|^2\dx t.
\end{eqnarray*}
\end{proof}
Following an idea of \cite{EK05}, for $\psi\in \mathcal{C}^\infty_{0}(\Omega)$ we let $\psi=J_{1}\phi$, with $\phi\in \mathcal{C}^\infty_{0}(\Omega)$
We deduce:
$$\int_{\Omega} |(-i\nabla+\A)\psi|^2-\la_{1}^\Dir(\omega)|\psi|^2\dx t=\int_{\Omega} J_{1}^2|(-i\nabla+\A)\phi|^2\dx t.$$
Let us now establish a lower bound for $\int_{\Omega} J_{1}^2|(-i\nabla+\A)\phi|^2\dx t$. Let us introduce a partition of unity:
$$\chi_{0,R}^2(s)+\chi_{1,R}^2(s)=1,$$
where $\chi_{j,R}(s)=\chi_{j}\left(R^{-1}s\right)$ with $\chi_{0}$ such that $\chi_{0}(s)=0$ for $s\in[-1/2,1/2]$ and $\chi_{0}(s)=1$ for $|s|\geq 1$.
We have:
$$\int_{\Omega} \frac{1}{1+s^2}|J_{1}\phi|^2\dx t\leq \int_{\Omega} \frac{1}{s^2}|J_{1}\chi_{0,R}\phi|^2\dx t+\int_{\Omega}|J_{1}\chi_{1,R}\phi|^2\dx t,$$
By the one dimensional Hardy inequality, we get:
$$\int_{\Omega} \frac{1}{s^2}|J_{1}\chi_{0,R}\phi|^2\dx t\leq 4\int_{\Omega} J_{1}^2(\dr_{s}|\chi_{0,R}\phi|)^2\dx t\leq 4 \int_{\Omega} J_{1}^2(\nabla|\chi_{0,R}\phi|)^2\dx t.$$
The diamagnetic inequality (see \cite[Chapter 2]{FouHel10}) implies that:
$$\int_{\Omega} J_{1}^2(\nabla|\chi_{0,R}\phi|)^2\dx t\leq \int_{\Omega} J_{1}^2|(-i\nabla+\A)\chi_{0,R}\phi|^2\dx t.$$
We infer:
$$\int_{\Omega} \frac{1}{1+s^2}|J_{1}\phi|^2\dx t\leq 4\int_{\Omega} J_{1}^2|(-i\nabla+\A)\chi_{0,R}\phi|^2\dx t+\int_{\Omega}|J_{1}\chi_{1,R}\phi|^2\dx t.$$
We apply \eqref{IMS-J0} to $\chi_{1,R}\phi$ and we get:
$$\int_{\Omega} J_{1}^2|(-i\nabla+\A)\chi_{1,R}\phi|^2\dx t=\int_{\Omega} |(-i\nabla+\A)J_{1}\chi_{1,R}\phi|^2\dx t-\la_{1}^\Dir(\omega)\int_{\Omega} J_{1}^2|\chi_{1,R}\phi|^2\dx t$$
and we deduce:
$$(\la_{1}^\Dir(\B,\Omega(R))-\la_{1}^\Dir(\omega))\int_{\Omega}|J_{1}\chi_{1,R}\phi|^2\dx t\leq \int_{\Omega} J_{1}^2|(-i\nabla+\A)\chi_{1,R}\phi|^2\dx t,$$
where $\la_{1}^\Dir(\B,\Omega(R))$ denotes the lowest eigenvalue of the magnetic Dirichlet Laplacian on $\Omega(R)$.
The diamagnetic inequality implies (see again \cite[Prop. 2.1.3]{FouHel10}) that we have a strict increasing of the energy in presence of a magnetic field ($R\geq R_{0}$): 
$$\la_{1}^\Dir(\B,\Omega(R))>\la_{1}^\Dir(0,\Omega(R))\geq\la_{1}^\Dir(\omega).$$ 
We infer:
$$\int_{\Omega}|J_{1}\chi_{1,R}\phi|^2\dx t\leq \left(\la_{1}^\Dir(\B,\Omega(R))-\la_{1}^\Dir(\omega)\right)^{-1}\int_{\Omega} J_{1}^2 |(-i\nabla+\A)\chi_{1,R}\phi|^2\dx t$$
so that:
\begin{multline*}
\int_{\Omega} \frac{1}{1+s^2}|J_{1}\phi|^2\dx t\leq 4\int_{\Omega} J_{1}^2|(-i\nabla+\A)\chi_{0,R}\phi|^2\dx t\\
+\left(\la_{1}^\Dir(\B,\Omega(R))-\la_{1}^\Dir(\omega)\right)^{-1}\int_{\Omega} J_{1}^2 |(-i\nabla+\A)\chi_{1,R}\phi|^2\dx t.
\end{multline*}
We deduce that:
\begin{equation*}
\int_{\Omega} \frac{1}{1+s^2}|J_{1}\phi|^2\dx t \leq \max\left(4,\left(\la_{1}^\Dir(\B,\Omega(R))-\la_{1}^\Dir(\omega)\right)^{-1}\right)\mathcal{Q}_{R,\A}(\phi),
\end{equation*}
where:
$$\mathcal{Q}_{R,\A}(\phi)=\int_{\Omega} J_{1}^2 |(-i\nabla+\A)\chi_{0,R}\phi|^2\dx t+\int_{\Omega} J_{1}^2 |(-i\nabla+\A)\chi_{1,R}\phi|^2\dx t.$$
Let us write a formula in the spirit of the so-called \enquote{IMS} formula (see \cite[Chapter 3]{CFKS87}):
\begin{eqnarray*}
&&\langle(-i\nabla+\A)J_{1}^2(-i\nabla+\A)\phi,\chi_{j,R}^2 \phi\rangle=\langle J_{1}^2(-i\nabla+\A)\phi,(-i\nabla+\A)\chi_{j,R}^2 \phi\rangle\\
&&=\langle\chi_{j,R} J_{1}^2(-i\nabla+\A)\phi,-i(\nabla\chi_{j,R})\phi\rangle+\langle J_{1}^2\chi_{j,R}(-i\nabla+\A)\phi,(-i\nabla+\A)(\chi_{j,R}\phi)\rangle.
\end{eqnarray*}
Then, we get:
\begin{multline*}
\langle\chi_{j,R} J_{1}^2(-i\nabla+\A)\phi,-i(\nabla\chi_{j,R})\phi\rangle+\langle J_{1}^2\chi_{j,R}(-i\nabla+\A)\phi,(-i\nabla+\A)(\chi_{j,R}\phi)\rangle\\
=\langle J_{1}^2(-i\nabla+\A)(\chi_{j,R}\phi),-i(\nabla\chi_{j,R})\phi\rangle+\langle J_{1}^2\chi_{j,R}(-i\nabla+\A)\phi,(-i\nabla+\A)(\chi_{j,R}\phi)\rangle.\\
-\|J_{1}\nabla\chi_{j,R}\phi\|^2.
\end{multline*}
We deduce that:
\begin{equation*}
\Re\left(\langle(-i\nabla+\A)J_{1}^2(-i\nabla+\A)\phi,\chi_{j,R}^2 \phi\rangle\right)=\int_{\Omega} |J_{1}(-i\nabla+\A)(\chi_{j,R}\phi)|^2\dx t-\|J_{1}\nabla\chi_{j,R}\phi\|^2
\end{equation*}
and thus:
\begin{equation*}
\Re\left(\langle(-i\nabla+\A)J_{1}^2(-i\nabla+\A)\phi, \phi\rangle\right)=\mathcal{Q}_{R,\A}(\phi)-\sum_{j=1}^2\|J_{1}\nabla\chi_{j,R}\phi\|^2.
\end{equation*}
We notice that:
$$\sum_{j=1}^2\|J_{1}\nabla\chi_{j,R}\phi\|^2\leq CR^{-2}\int_{\Omega(R)}|J_{1}\phi|^2\dx t.$$
Moreover, by using Lemma \ref{IMS-J0} and the min-max principle, we find:
$$\la_{1}^{\Dir,\Neu}(B,\Omega(R))\int_{\Omega(R)} |J_{1}\phi|^2 \dx t\leq\int_{\Omega(R)} J_{1}^2|(-i\nabla+\A)\phi|^2\dx t+\la_{1}^\Dir(\omega)\int_{\Omega(R)} J_{1}^2|\phi|^2\dx t$$
so that:
$$\int_{\Omega(R)}|J_{1}\phi|^2\dx t\leq\left(\la_{1}^{\Dir,\Neu}(B,\Omega(R))-\la_{1}^\Dir(\omega))\right)^{-1}\int_{\Omega(R)} J_{1}^2 |(-i\nabla+\A)\phi|^2\dx t.$$
Since we have $\la_{1}^\Dir(B,\Omega(R))\geq \la_{1}^{\Dir,\Neu}(B,\Omega(R))$, we conclude that:
\begin{multline*}
\int_{\Omega} \frac{1}{1+s^2}|J_{1}\phi|^2\dx t\\
 \leq(1+CR^{-2})\max\left(4,\left(\la_{1}^{\Dir,\Neu}(\B,\Omega(R))-\la_{1}^\Dir(\omega)\right)^{-1}\right)\int_{\Omega} J_{1}^2 |(-i\nabla+\A)\phi|^2\dx t.
\end{multline*}

\subsection{Proof of Proposition~\ref{perturbed}}
This section is devoted to the proof of Proposition \ref{perturbed}. We can write:
$$D\Phi_{\eps}=\Id+\mathcal{E}(s,t),$$
where $\mathcal{E}$ is smooth and compactly supported in $K$ and satisfies $|\mathcal{E}|\leq C\eps$. We deduce that the metrics $G_{\eps}$ induced by $\Phi_{\eps}$ satisfies:
$$G_{\eps}^{-1}=\Id+\tilde{\mathcal{E}}(s,t),$$
where $\tilde{\mathcal{E}}$ is smooth and compactly supported in $K$ and such that $|\tilde{\mathcal{E}}|\leq C\eps$. The new vector potential becomes $\mathcal{A}_{\eps}=\A+\hat{\mathcal{E}}(s,t)$. Let us introduce a smooth cutoff function $\chi$ being $1$ on $K$.
The quadratic form on the perturbed tube is given by:
$$Q_{\mathcal{A}_{\eps},\eps}(\psi)=\int_{\Omega} \langle G_{\eps}^{-1}(-i\nabla+\mathcal{A}_{\eps})\psi,(-i\nabla+\mathcal{A}_{\eps})\psi\rangle\,|g_{\eps}|^{1/2}\dx t$$
and we get, for $\psi\in\mathcal{C}_{0}^\infty(\Omega)$:
\begin{multline*}
Q_{\mathcal{A}_{\eps},\eps}(\psi)-\la_{1}^\Dir(\omega)\int_{\Omega}|\psi|^2 |g_{\eps}|^{1/2}\dx t\\
\geq\int_{\Omega} |(-i\nabla+\mathcal{A}_{\eps})\psi|^2\dx t-\la_{1}^\Dir(\omega)\int_{\Omega}|\psi|^2\dx t-C\eps\int_{\Omega}\chi^2 |(-i\nabla+\mathcal{A}_{\eps})\psi|^2\dx t-C\eps\int_{\Omega}\chi^2|\psi|^2\dx t.
\end{multline*}
This can be rewritten in the form:
\begin{eqnarray*}
&Q_{\mathcal{A}_{\eps},\eps}(\psi)-\la_{1}^\Dir(\omega)\int_{\Omega}|\psi|^2 |g_{\eps}|^{1/2}\dx t\geq&\int_{\Omega} |(-i\nabla+\mathcal{A}_{\eps})\psi|^2\dx t-\la_{1}^\Dir(\omega)\int_{\Omega}|\psi|^2\dx t\\
&&-C\eps\left(\int_{\Omega}\chi^2 |(-i\nabla+\mathcal{A}_{\eps})\psi|^2\dx t-\la_{1}^\Dir(\omega)\int_{\Omega}|\chi\psi|^2\dx t\right)\\
&&-\tilde C\eps\int_{\Omega}\chi^2|\psi|^2\dx t.
\end{eqnarray*}
Then, the following identity holds:
\begin{eqnarray*}
\int_{\Omega}\chi^2 |(-i\nabla+\mathcal{A}_{\eps})\psi|^2\dx t=\int_{\Omega}|(-i\nabla+\mathcal{A}_{\eps})(\chi\psi)|^2\dx t-\int_{\Omega} |(\nabla\chi)\psi|^2\,dt+\frac{1}{2}\int_{\Omega} \Delta(\chi^2)|\psi|^2\dx t.
\end{eqnarray*}
We deduce that:
$$Q_{\mathcal{A}_{\eps},\eps}(\psi)-\la_{1}^\Dir(\omega)\int_{\Omega}|\psi|^2 |g_{\eps}|\dx t\geq q_{\eps}(\psi)-C\eps q_{\eps}(\chi\psi)-C\int_{\Omega}\underline{\chi}^2|\psi|^2\dx t,$$
where $\underline{\chi}$ is a smooth cutoff function supported on a compact slightly bigger than $K$ and where $q_{\eps}$ is defined by:
$$q_{\eps}(\psi)=\int_{\Omega} |(-i\nabla+\mathcal{A}_{\eps})\psi|^2\dx t-\la_{1}^\Dir(\omega)\int_{\Omega}|\psi|^2\dx t.$$
Writing $\psi=J_{1}\varphi$ with $\varphi\in\mathcal{C}^\infty_{0}(\Omega)$, we recall that:
$$q_{\eps}(\chi\psi)=\int_{\Omega} J_{1}^2|(-i\nabla+\mathcal{A}_{\eps})\chi\varphi|^2\dx t.$$
We infer an upper bound in the form:
$$q_{\eps}(\chi\psi)\leq 2\int_{\Omega} J_{1}^2|(-i\nabla+\mathcal{A}_{\eps})\varphi|^2\dx t+C\int_{\Omega}\underline{\chi}^2 |\varphi|^2\dx t=2q_{\eps}(\psi)+C\int_{\Omega}\underline{\chi}^2 |\varphi|^2\dx t.$$
We deduce:
$$Q_{\mathcal{A}_{\eps},\eps}(\psi)-\la_{1}^\Dir(\omega)\int_{\Omega}|\psi|^2 |g_{\eps}|\dx t\geq (1-C\eps)q_{\eps}(\psi)-C\eps\int_{\Omega}\underline{\chi}^2|\psi|^2\dx t.$$
We again notice that:
$$q_{\eps}(\psi)=\int_{\Omega} J_{1}^2|(-i\nabla+\mathcal{A}_{\eps})\varphi|^2\dx t$$
so that:
$$q_{\eps}(\psi)\geq (1-\eps) \int_{\Omega} J_{1}^2|(-i\nabla+\A)\varphi|^2\dx t-C\eps\int_{\Omega} \chi^2|J_{1}\varphi|^2\dx t.$$
We get:
\begin{multline*}
Q_{\mathcal{A}_{\eps},\eps}(\psi)-\la_{1}^\Dir(\omega)\int_{\Omega}|\psi|^2 |g_{\eps}|\dx t\geq(1-C\eps) \left(\int_{\Omega}|(-i\nabla+\A)\psi|^2\dx t-\la_{1}^\Dir(\omega)\int_{\Omega}|\psi|^2\dx t\right)\\
-C\eps\int_{\Omega} \underline{\chi}^2|\psi|^2\dx t.
\end{multline*}
We use the magnetic Hardy inequality:
$$\int_{\Omega}|(-i\nabla+\A)\psi|^2\dx t-\la_{1}^\Dir(\omega)\int_{\Omega}|\psi|^2\dx t\geq\int_{\Omega}\frac{c_{R}(\B)}{1+s^2}|\psi|^2\dx t.$$
to infer, for $\eps$ small enough:
$$Q_{\mathcal{A}_{\eps},\eps}(\psi)-\la_{1}^\Dir(\omega)\int_{\Omega}|\psi|^2 |g_{\eps}|\dx t\geq 0.$$

\subsection{Proof of Proposition~\ref{strong-b}}
In this section we prove Proposition \ref{strong-b}. We first write:
$$\mathfrak{Q}^{[d]}_{1,b\A}(\psi)=\int_{\R\times\omega} G^{-1}(-i\nabla+b\mathcal{A})\psi\cdot (-i\nabla+b\mathcal{A})\psi\, |g|^{1/2}\dx t.$$
We split this integral into two parts:
\begin{eqnarray*}
&\mathfrak{Q}^{[d]}_{1,b\A}(\psi)=&\int_{\Omega(R_{0})} G^{-1}(-i\nabla+b\mathcal{A})\psi\cdot (-i\nabla+b\mathcal{A})\psi\, |g|^{1/2}\dx t\\
&                                                         &+\int_{\Omega\setminus \Omega(R_{0})} G^{-1}(-i\nabla+b\mathcal{A})\psi\cdot (-i\nabla+b\mathcal{A})\psi\, |g|^{1/2}\dx t
\end{eqnarray*}
Since the curvature is zero on $\Omega\setminus \Omega(R_{0})$, we have:
$$\int_{\Omega\setminus \Omega(R_{0})} G^{-1}(-i\nabla+b\mathcal{A})\psi\cdot (-i\nabla+b\mathcal{A})\psi\, |g|^{1/2}\dx t=\int_{\Omega\setminus \Omega(R_{0})} |(-i\nabla+b\A)\psi|^2\dx t.$$
and the diamagnetic inequality implies that:
$$\int_{\Omega\setminus \Omega(R_{0})} |(-i\nabla+b\A)\psi|^2\dx t\geq \la_{1}^\Dir(\omega)\int_{\Omega\setminus \Omega(R_{0})}|\psi|^2\dx t=\la_{1}^\Dir(\omega)\int_{\Omega\setminus \Omega(R_{0})}|\psi|^2\,|g|^{1/2}\dx t.$$
Moreover we have:
$$\int_{\Omega(R_{0})} G^{-1}(-i\nabla+b\mathcal{A})\psi\cdot (-i\nabla+b\mathcal{A})\psi\, |g|^{1/2}\dx t\geq\la_{1}^{\Dir,\Neu}(b\B, \Omega(R_{0}))\int_{\Omega(R_{0})}|\psi|^2\,|g|^{1/2}\dx t,$$
where $\la_{1}^{\Dir,\Neu}(b\B, \Omega(R_{0}))$ is the lowest eigenvalue of the magnetic Laplacian $(-i\nabla+b\A)^2$ defined on $\Phi(\Omega(R_{0}))$ with Neumann condition on $\Phi(|s|=R_{0})$. Since the magnetic field does not vanish on $\Omega(R_{0})$, it is standard to establish (see \cite[Section 1.4.3]{FouHel10}) that there exists $c, b_{0}>0$ such that for $b\geq b_{0}$, we have:
$$\la_{1}^{\Dir,\Neu}(b\B, \Omega(R_{0}))\geq c b\inf_{x\in\Phi(\Omega(R_{0}))}\|\B(x)\|,$$
where $\|\B(x)\|$ is the norm of $\B$ defined in \cite[Section 1.4.3]{FouHel10}.
For $b$ such that we have $\displaystyle{c b\inf_{x\in\Phi(\Omega(R_{0}))}\|\B(x)\|\geq \la^\Dir_{1}(\omega)}$, we get the conclusion.

\paragraph{Acknowledgments.} 
The authors would like to dedicate this paper 
to the memory of Pierre Duclos (1948--2010)
who stimulated their interest in spectral theory of quantum waveguides.
The first author has been partially supported 
by RVO61389005 and the GACR grant No.\ P203/11/0701.
The second author would like to thank the Mittag-Leffler Institute 
where part of this paper was written.

\bibliographystyle{mnachrn}
\bibliography{BIB}
\end{document}